\documentclass[final]{siamltex}

\usepackage[T1]{fontenc}
\usepackage{amssymb}
\usepackage{graphics}
\usepackage{graphicx}
\usepackage{epsf}
\usepackage{float}

\newcommand{\R}{\mathbb{R}}
\newcommand{\N}{\mathbb{N}}
\newcommand{\D}{\mathbb{D}}
\newcommand{\1}{\mathbf{1}}
\newcommand{\F}{\mathcal{F}}

\newcommand{\preuve}[2]{{\raggedleft \bf Proof #1.} #2  \begin{flushright} \qquad\endproof \end{flushright}}

\begin{document}

\bibliographystyle{ieeeTr}

\title{American Options Based on Malliavin Calculus and
Nonparametric Variance Reduction Methods}

\author{L.~A. Abbas-Turki\thanks{Université Paris-Est, Laboratoire d'Analyse et de Mathématiques Appliquées, Champs-sur-Marne, 77454 Marne-la-Vallée Cedex2, France.} \and  B.~Lapeyre\thanks{Ecole des Ponts ParisTech, CERMICS Applied Probability Research Group, Champs-sur-Marne, 77455 Marne-la-Vallée cedex 2, France.}}

\maketitle

\begin{abstract}
This paper is devoted to pricing American options
using Monte Carlo and the Malliavin calculus. Unlike the majority of
articles related to this topic, in this work we will not use
localization fonctions to reduce the variance. Our method is based
on expressing the conditional expectation $E[f(S_{t})/S_{s}]$ using the
Malliavin calculus without localization. Then the variance of the estimator of
$E[f(S_{t})/S_{s}]$ is reduced using closed formulas, techniques based
on a conditioning and a judicious choice of the number of simulated paths.
Finally, we perform the stopping times version of
the dynamic programming algorithm to decrease the bias. On the one hand,
we will develop the Malliavin calculus tools for exponential multi-dimensional diffusions
that have deterministic and no constant coefficients. On the other hand, we will
detail various nonparametric technics to reduce the variance.
Moreover, we will test the numerical efficiency of our method
on a heterogeneous CPU/GPU multi-core machine.
\end{abstract}

\begin{keywords}
American Options, Malliavin Calculus, Monte Carlo, GPU.
\end{keywords}

\begin{AMS}
60G40, 60H07
\end{AMS}

\pagestyle{myheadings}
\thispagestyle{plain}

\section*{Introduction and objectives}

To manage CPU (Central Processing Unit) power dissipation, the
processor makers have oriented their architectures to multi-cores.
This switch in technology led us to study the pricing algorithms
based on \textbf{M}onte \textbf{C}arlo (\textbf{MC}) for multi-core architectures using CPUs
and GPUs (Graphics Processing Units) in \cite{lok} and \cite{lokch}.
In the latter articles we basically studied the impact of using
GPUs instead of CPUs for pricing European options using MC and
American options using the \textbf{L}ongstaff and \textbf{S}chwartz
(\textbf{LS}) algorithm~\cite{long}.
The results of this study proves that we can greatly decrease the
execution time and the energy consumed during the simulation.

In this paper, we explore another method to price \textbf{A}merican \textbf{O}ptions
(\textbf{AO}) and which is based on \textbf{MC} using the \textbf{M}alliavin calculus
(\textbf{MCM}). Unlike the LS method that uses a regression phase which is difficult
to parallelize according to \cite{lokch}, the MCM is a square\footnote{What
we mean by square Monte Carlo is not necessarily simulating a square number
of trajectories, but a Monte Carlo simulation that requires a Monte Carlo
estimation, for each path, of an intermediate value (here the continuation)
and this can be done by using the same set of trajectories as the first Monte
Carlo simulation.} Monte Carlo method which is more adapted to
multi-cores than the LS method. Moreover, using MCM
without localization does not depend on parametric regression,
we can increase the dimensionality of the problem without any
constraints except for adding more trajectories if we aim at more
accurate results.

American contracts can be exercised at any trading date until
maturity and their prices are given, at each time $t$, by
\cite{aa}
\begin{eqnarray}
\begin{array}{c}
P_{t}(x) = \sup_{\theta \in \mathcal{T}_{t,T}} E_{t,x} \left(
e^{-r(\theta -t)} \Phi ( S_{\theta } ) \right), \label{Amer}
\end{array}
\end{eqnarray}
\noindent where $\mathcal{T}_{t,T}$ is the set of stopping times
in the time interval $[t,T]$, $E_{t,x}$ is the expectation
associated to the risk neutral probability knowing that $S_{t}=x$
and $r$ and $\phi(S_{t})$ are respectively the risk neutral interest
rate and the payoff of the contract.

With Markovian models (which is the case in this article), to
evaluate numerically the price (\ref{Amer}), we first need to approach stopping
times in $\mathcal{T}_{t,T} $ with stopping times taking values in
the finite set $t=t_{0}<t_{1}<...<t_{n}=T$. When we do this
approximation, pricing American options can be reduced to the
implementation of the dynamic programming algorithm \cite{aa}.
Longstaff and Schwartz consider the stopping times formulation of
the dynamic programming algorithm which allows them to reduce the
bias by using the actual realized cash flow. We refer the reader
to \cite{lamb} for a formal presentation of the LS algorithm and
details on the convergence. In (\ref{lon}), we rewrite the dynamic
programming principle in terms of the optimal stopping times $\tau_k$,
for each path, as follows
\begin{eqnarray}
\begin{array}{c}
\tau_n=T,\\
\forall k \in \{ n-1,...,0 \},\quad \tau_k = t_k 1_{A_k} + \tau_{k+1} 1_{A^{c}_k}, \label{lon}
\end{array}
\end{eqnarray}
where the set $A_k = \{ \Phi(S_{t_{k}}) > C(S_{t_{k}}) \}
$ and $C(S_{t_{k}})$ is the continuation value whose expression is
given by
\begin{eqnarray}
C(S_{t_{k}}) = E \left( e^{-r(t_{k+1}-t_{k})}
P_{t_{k+1}}(S_{t_{k+1}}) \Big| S_{t_{k}} \right). \label{contcont}
\end{eqnarray}
Thus, to evaluate the price (\ref{Amer}), we need to
estimate $C(S_{t_{k}})$. Algorithms devoted to American pricing
and based on Monte Carlo, differ essentially in the way they
estimate and use the continuation value (\ref{contcont}). For
example the authors of \cite{tsi} perform a regression to estimate
the continuation value, but unlike \cite{long}, they use
$C(S_{t_{k}})$ instead of the actual realized cash flow
$P_{t_{k+1}}(S_{t_{k+1}})$ to update the price in (\ref{lon}).
Other methods use the Malliavin Calculus \cite{zan} or the quantization
method \cite{ball} for $C(S_{t_{k}})$ estimation. In
\cite{lokch}, we implement the LS method because it is gaining
widespread use in the financial industry. As far as this work
is concerned, we are going to implement MCM but unlike \cite{zan}
we use the induction (\ref{lon}) for the implementation and we
reduce the variance differently, without using localization.

Formally speaking, if $r=0$, we can rewrite the continuation
using the Dirac measure $\varepsilon_{x}(\cdot )$ at the point $x$
\begin{eqnarray}
C(x)  = \frac{E \left(
P_{t_{k+1}}(S_{t_{k+1}}) \varepsilon_{x}(S_{t_{k}} )
 \right)}{E \left(
\varepsilon_{x}(S_{t_{k}} ) \right)} = \frac{E \left(
P_{t_{k+1}}(S_{t_{k+1}}) 1_{S_{t_{k}} \geq x}(S_{t_{k}}) \pi_{t_{k},t_{k+1}}
 \right)}{E \left( 1_{S_{t_{k}} \geq x } (S_{t_{k}} ) \pi_{t_{k},t_{k+1}} \right)}.
\label{condirnot}
\end{eqnarray}
\noindent The second equality is obtained using the Malliavin calculus and we will specify,
in section \ref{sec2} expression (\ref{pitk}), the value of $\pi_{t_{k},t_{k+1}}$ by an integration by part argument for
the \textbf{M}ulti-dimensional \textbf{E}xponential \textbf{D}iffusions with deterministic
\textbf{C}oefficients (\textbf{MEDC}) model
\begin{eqnarray*}
dS_t = S_t \sigma (t) dW_t,\quad S_0 = y,
\end{eqnarray*}
in the case of deterministic non-constant triangular matrix $ \sigma (t) $ and when
$\sigma_{ij} (t) = \sigma_{ij} \delta (i-j)$ with a fixed constant $\sigma_{ij}$ (The latter case will be used as a benchmark).
Instead of simulating directly the last term in (\ref{condirnot}), in section \ref{sec3} we project
$1_{S_{t_{k}} \geq x } (S_{t_{k}} ) \pi_{t_{k},t_{k+1}} $ using a conditioning as follows
\begin{eqnarray}
C(x)   = \frac{E \left(
P_{t_{k+1}}(S_{t_{k+1}}) E \left[ 1_{S_{t_{k}} \geq x}(S_{t_{k}})  \pi_{t_{k},t_{k+1}} \big| \{ \int_{0}^{t_{k+1}}
\sigma_{ij} (u) dW^j_u \}_{1 \leq  j \leq i \leq d} \right] \right)}{E \left( E \left[ 1_{S_{t_{k}} \geq x } (S_{t_{k}} )
\pi_{t_{k},t_{k+1}} \big| \{ \int_{0}^{t_{k+1}} \sigma_{ij} (u) dW^j_u \}_{1 \leq  j \leq i \leq d} \right] \right)}.
\label{condir}
\end{eqnarray}
Then, in section \ref{sec4}, we estimate (\ref{condir}) by Monte Carlo simulation and we use the approximation
\begin{eqnarray}
C(x) \approx \frac{ \frac{1}{N'} \sum_{l=1}^{N'}
P^{l}_{t_{k+1}}(S_{t_{k+1}}) h(x, \{ \int_{0}^{t_{k+1}}
\sigma_{ij} (u) dW^j_u \}_{1 \leq  j \leq i \leq d}^{l})}
{ \frac{1}{N} \sum_{l=1}^{N} h(x, \{ \int_{0}^{t_{k+1}}
\sigma_{ij} (u) dW^j_u \}_{1 \leq  j \leq i \leq d}^{l})},
\label{condapp}
\end{eqnarray}
$ h(x,\{ y_{ij} \}_{j \leq i}) =  E( 1_{S_{t_{k}} \geq x } (S_{t_{k}} ) \pi_{t_{k},t_{k+1}} \big| \{ \int_{0}^{t_{k+1}}
\sigma_{ij} (u) dW^j_u \}_{1 \leq  j \leq i \leq d} = \{ y_{ij} \}_{1 \leq  j \leq i \leq d})$ and $N \neq N'$. Thus, in section \ref{sec4},
we provide another method to accelerate the convergence based on a choice of the appropriate relation between
$N$ and $N'$ that reduces the variance of the quotient (\ref{condapp}). Note that, even if one can reduce the variance by an "appropriate" control variable, we choose here not to implement this kind of method because it is not standard for American options.

In the last section, on the one hand, we provide the numerical result
comparison of LS and MCM. On the other hand, we study
the results of using the two variance reduction methods (\ref{condir}) and (\ref{condapp}).
Finally, we test the parallel capabilities of MCM on a desktop computer that has the following
specifications: Intel Core i7 Processor 920 with 9GB of tri-channel memory at frequency 1333MHz.
It also contains one NVIDIA GeForce GTX 480.

Let us begin with section \ref{sec1} in which we establish the notations, the
Malliavin calculus tools and the model used.

\section{Notations, hypothesis and key tools\label{sec1}}

Let $T$ be the maturity of the American contract, $(\Omega ,
\F, P)$ a probability space on which we define an
$d$-dimensional standard Brownian motion $W=(W^{1},...,W^{d})$ and
$\mathbb{F}=\{ \F_{s} \}_{s \leq T}$ the
$P$-completion of the filtration generated by $W$ until
maturity. Moreover, we denote by $\{ \F_{s}^{i,...,d} \}_{s \leq t}$
the $P$-completion of the filtration generated by $(W^{i},...,W^{d})$
until the fixed time $t \in [0,T]$. The process $S_{t}$ models the price of a vector of
assets $S_{t}^1,...,S_{t}^d$ which constitute the solution of the following
stochastic differential equation ( ' is the transpose operator)
\begin{eqnarray}
\frac{dS_{t}^i}{S_{t}^i} = (\sigma_{i}(t) )' dW_t,\quad S_{0}^i = z_{i} ,\quad i=1,..,d,
\label{Dyn}
\end{eqnarray}
\noindent where $\sigma (t) = \{ \sigma_{ij}(t) \}_{1 \leq i,j \leq d}$ is a deterministic triangular matrix
($\{ \sigma_{ij} (t) \}_{ i < j } = 0$). We suppose that the matrix $\sigma (t)$ is invertible, bounded and
uniformly elliptic which insures the existence of the inverse matrix $\rho (t) = \sigma^{-1} (t)$ and
its boundedness.

We choose the dynamic (\ref{Dyn}) because it is largely used
for equity models, HJM interest rate models and variance swap models. Moreover, the case of
$\sigma_{ij} (t) = \sigma_{ij} \delta (i-j)$ ($\sigma_{ij}$ is a constant) will be easily tested
in the section \ref{sec5}. One should notice that in the case where the dynamic of $S$ is given by
\begin{eqnarray*}
\frac{dS_{t}^i}{S_{t}^i} = (\sigma_{i} )'(t,S_t) dW_t,\quad S_{0}^i = z_{i} ,\quad i=1,..,d,
\end{eqnarray*}
we can use, for instance, the following Euler scheme to reduce this model to the model (\ref{Dyn})
\begin{eqnarray*}
d \log ( S_{t}^i ) = \sum_{k=0}^{n-1} 1_{t \in [t_k, t_{k+1} [ } \left[ (\sigma_{i} )' (t_k,S_{t_k})  dW_t
- \frac{1}{2} [ (\sigma_{i} )' \sigma_{i} ] (t_k,S_{t_k})  dt  \right],
\\ S_{0}^i = z_{i} ,\quad i=1,..,d ,\quad  t_k = \frac{kT}{n}. \hspace{20mm}
\end{eqnarray*}
Note that this scheme does not discretize the process $S$ but the process $\log (S)$.

Throughout this article, we will use two operators: The Malliavin derivatives $D$ and the Skorohod integral $\delta $
and we define them in the same way as in \cite{vlad}. For a fixed $m \in \N$, we define the subdivision $\{ t_{m}^{k} \}_{k \leq 2^{m}}$
of the finite interval $[0,T]$ by: $t_{m}^{k} = kT/2^{m}$. Then we introduce $\mathcal{S} (\R^{2^{m}} )$ the Schwartz space of infinitely
differentiable and rapidly decreasing functions on $\R^{2^{m}}$. Let $f \in \mathcal{S} (\R^{2^{m}} )$, we define the set $\mathfrak{S}^m$ of simple
functionals by the following representation
\begin{eqnarray*}
F \in \mathfrak{S}^m \Leftrightarrow F = f \left(W_{t_{m}^{1}}-W_{t_{m}^{0}}, W_{t_{m}^{2}}-W_{t_{m}^{1}},..., W_{t_{m}^{2^{m}}}-W_{t_{m}^{2^{m}-1}} \right).
\end{eqnarray*}
One can prove that $\mathfrak{S} = \bigcup_{m \in \N } \mathfrak{S}^m$ is a linear and dense subspace in $L^{2}(\Omega )$ and that the Malliavin
derivatives $D^{i}F$ of $F \in \mathfrak{S}$ defined by
\begin{eqnarray*}
D^{i}_{t}F = \sum_{k=0}^{2^{m}-1} \frac{\partial f}{\partial x^{i,k}} \left(W_{t_{m}^{1}}-W_{t_{m}^{0}},..., W_{t_{m}^{2^{m}}}-W_{t_{m}^{2^{m}-1}} \right)
\1_{[t_{m}^{k}, t_{m}^{k+1} [} (t)
\end{eqnarray*}
represents a process of $L^{2}(\Omega \times [0,T])$ with values in $L^{2}( [0,T])$. We associate to $\mathfrak{S} $ the norm $|| \cdot ||_{1,2}$ defined by
\begin{eqnarray*}
||F||_{1,2}^{2} = E |F|^{2} + \sum_{i=1}^{d} E \int^{T}_{0} (D^{i}_{t}F)^{2} dt.
\end{eqnarray*}
Finally, the space $\D^{1,2}$ is the closure of $\mathfrak{S} $ with respect to this norm and we say that
$F \in \D^{1,2}$ if there exists a sequence $F_m \in \mathfrak{S}$ that converges to $F$ in $L^{2}(\Omega )$ and
that $D_u F_m $ is a Cauchy sequence in $L^{2}(\Omega \times [0,T]) $.

Now we use the duality property between $\delta $ and $D $ to define the Skorohod integral $\delta $.
We say that the process $U \in Dom (\delta )$ if $\forall F \in \D^{1,2}$
\begin{eqnarray*}
E \left( \int^{T}_{0} U_{t} \cdot D_{t}F dt \right) \leq C(U) ||F||_{1,2},
\end{eqnarray*}
where $C(U)$ is a positive constant that depends on the process $U$.
If $U \in Dom (\delta )$, we define the Skorohod integral $\delta (U) = \int U_t \delta W_t $ by
\begin{eqnarray}
\forall F \in \D^{1,2} ,\quad E \left( F \int^{T}_{0} U_t \cdot \delta W_t \right) = E \left( F \delta (U) \right)
= E \left( \int^{T}_{0} U_{t} \cdot D_{t}F dt \right),
\label{dual1}
\end{eqnarray}
$(\cdot )$ is the inner scalar product on $\R^d$.\\

\noindent Below, we give some standard properties of the operators $D$ and $\delta $:
\begin{itemize}
\item[1.] If the process $U_t$ is adapted, $\delta (U) = \int U_t \delta W_t $ coincides with the Itô integral
$\int U_t dW_t$.

\item[2.] \textit{The Chain Rule:} Let $F=(F_1,F_2,...,F_k) \in (\D^{1,2})^k$ and $\phi : \R^k \rightarrow \R$
a continuously differentiable function with bounded partial derivatives. Then $\phi (F_1,F_2,...,F_k) \in \D^{1,2}$
and:
\begin{eqnarray*}
D_{t} \phi (F_1,F_2,...,F_k) = \sum_{i=1}^{k} \frac{\partial \phi}{\partial x^{i}} (F_1,F_2,...,F_k) D_{t} F_i.
\end{eqnarray*}
\item[3.] \textit{The Integration by Parts:} The IP formula will be extensively
used in the next section on the time intervals $I=(0,s)$ and $I=(s,t)$ with $s < t \in ]0,T]$: Let $F \in \D^{1,2}$,
an adapted process $U \in Dom ( \delta )$ and that $FU \in Dom ( \delta )$.
For each $1 \leq i \leq d$ we have the following equality
\begin{eqnarray}
\int_{I} F U_u \delta^i W_u = F \int_{I} U_u dW^i_u - \int_{I} U_u D^i_u F du.
\label{IP1}
\end{eqnarray}
\end{itemize}
To simplify the notations, we denote $H_i(S_{s}^i) = H(S_{s}^i - x_i)$
for the heaviside function of the difference between the $i^{th}$ stock and the
$i^{th}$ coordinate of the positive vector $x$.

Throughout this article, we will suppose that $g \in \mathcal{E}_{b}$ is a measurable
function with polynomial growth
\begin{eqnarray*}
\mathcal{E}_{b}(\R^{d})=\left\{f \in \mathcal{M}(\R^d)
\hspace{1mm} : \hspace{2mm}\exists C>0 \hspace{2mm} and
\hspace{2mm} m \in \mathbb{N}; \hspace{2mm} f(y) \leq C(1+|y|_d)^{m})
\right\},
\end{eqnarray*}
where $ \mathcal{M}(\R^d) $ is the set of measurable functions on $\R^d$. The elements of the set $\mathcal{E}_{b}(\R^{d})$ satisfy the finiteness of the expectations computed in this article.

\section{The expression of the continuation value \label{sec2}}

The first theorem of this section provides the expression
of the continuation (\ref{contcont}) when using Malliavin
calculus for MEDC models. This theorem can be considered
as an extension of the log-normal multi-dimensional model
detailed in \cite{zan}. In Theorem \ref{theoExp}, we provide
the expression of $\Gamma^{k}_{s,t}$, introduced in Theorem \ref{theo1},
without using Malliavin derivatives and this expression can be
computed using the relation (\ref{recdel}). The last theorem is
a special case of the first one because we take $\sigma_{ij} (t) = \sigma_{ij} \delta (i-j)$
($\sigma_{ij}$ is a constant) that will be used to test
numerically our nonparametric variance reduction methods
detailed in section \ref{sec3} and section \ref{sec4}.

\begin{theorem}
\label{theo1}
For any $s \in ]0,t[$, $g \in \mathcal{E}_{b}$ and $x>0$
\begin{eqnarray}
E \left( g(S_{t}) \Big| S_{s} = x \right) =
\frac{T_{s,t}[g](x)}{T_{s,t}[1](x)}, \label{condp}
\end{eqnarray}
with
\begin{eqnarray}
T_{s,t}[f ](x)=E \left( f(S_{t}) \Gamma_{s,t} \prod_{k=1}^{d}
\frac{H_k(S_{s}^k)}{S_{s}^k}  \right),  \label{TT}
\end{eqnarray}
where $\Gamma_{s,t} = \Gamma_{s,t}^1$ and $\Gamma_{s,t}^1$ can be computed by the following induction scheme
$\Gamma_{s,t}^d = \pi_{s,t}^{d,d}$, for $ k \in \{ 1,...,d-1 \}\hspace{-1mm}:
\hspace{1mm} \Gamma_{s,t}^k = \Gamma_{s,t}^{k+1} \pi_{s,t}^{k,d} - \sum_{j=k+1}^d \int_{0}^{t} D_u^j \Gamma^{k+1}_{s,t} D_u^j \pi^{k,d}_{s,t} du $ with
\begin{eqnarray*}
\pi_{s,t}^{k,d} = 1 + \sum_{j=k}^{d} \int_{0}^{t} \varphi_{jk}(u) dW^{j}_u
,\quad \varphi_{jk}(u) = \frac{1}{s} \rho_{jk}(u) 1_{u\in ]0,s[}  - \frac{1}{t-s}\rho_{jk}(u) 1_{u\in ]s,t[},
\end{eqnarray*}
where $\rho $ is the inverse matrix $\rho (u) = \sigma^{-1}(u)$.
\end{theorem}

From this theorem the value of $\pi_{t_{k},t_{k+1}}$ in (\ref{condirnot}) is given by
\begin{eqnarray}
\pi_{t_{k},t_{k+1}} = \Gamma_{t_{k},t_{k+1}} \prod_{i=1}^{d}
\frac{1}{S_{t_{k}}^i}.
\label{pitk}
\end{eqnarray}

To prove Theorem \ref{theo1}, we need the following two lemmas which are proved in the appendix. Lemma \ref{lem2} expresses the independence of the sum  $\sum_{i=k}^{d} \rho_{ik}(u) D_{u}^{i} g(S_{t})$ from the variable $u$.
\begin{lemma}
For any $u \in ]0,t[$ and $f \in \mathcal{C}^1(\R^d) $ then
\begin{eqnarray}
\sum_{i=k}^{d} \rho_{ik}(u) D_{u}^{i} f(S_{t}) = S_{t}^k \partial_{x_k} f(S_{t}),\quad \rho (u) = \sigma^{-1} (u).
\label{SumSimp}
\end{eqnarray}
\label{lem2}
\end{lemma}

The following lemma constitutes with equality (\ref{SumSimp}) the two keys of the proof of Theorem \ref{theo1}.

\begin{lemma}
For any $I \subset ]0,t[$, $h \in \mathcal{C}_b^{\infty }(\R )$, $x \in \R^d_+$ and $F \in \D^{1,2}$, we have
\begin{eqnarray}
\begin{array}{ccc}
E \left(\int_{I} \frac{F D_{u}^{k} h(S_{s}^k)}{\sigma_{kk} (u)} du  \right) &
 =  & E \left( h(S_{s}^k) F \sum_{i=k}^{d} \int_{I} \rho_{ik}(u) d W^{i}_u \right) \\ & - &
E \left( h(S_{s}^k) \sum_{i=k}^{d} \int_{I} \rho_{ik}(u) D_{u}^{i} F du \right),
\end{array}
\label{rel1}
\end{eqnarray}
where $\rho $ is the inverse matrix $\rho (u) = \sigma^{-1}(u)$.
\label{lem1}
\end{lemma}

\preuve{of Theorem \ref{theo1}}{
To prove Theorem \ref{theo1}, it is sufficient to prove the following recursive relation
on the parameter $k$ for each $h_i \in \mathcal{C}^{\infty }_b(\R )$ and $f \in \mathcal{C}^1(\R^d) \cap \mathcal{E}_{b}(\R^{d})$
\begin{eqnarray}
E \left( f(S_{t}) \prod_{i=1}^{d} h'_i(S_{s}^i)  \right) =E \left( f(S_{t}) \Gamma^{k+1}_{s,t} \prod_{i=1}^{k}
h'_i(S_{s}^i) \prod_{i=k+1}^{d} \frac{h_i(S_{s}^i)}{S_{s}^i}  \right).  \label{Induc}
\end{eqnarray}
Indeed, if it is the case then by density of $ \mathfrak{S}$ in $L^{2}(\Omega )$, one can approximate
$f(S_{t}) \in L^{2}(\Omega )$ by $F_m \in \mathfrak{S}$ and pass to the limit on the left and on the right
term of (\ref{Induc}) using Cauchy-Schwarz inequality and the dominated convergence theorem.
Let us now consider the singularity due to the heaviside, let $\phi \in \mathcal{C}^{\infty }_c(\R )$ be a mollifier function with
support equal to $[-1, 1]$ and such that $\int_{\R } \phi (y) dy =1$, then for any $y \in \R $ we define
\begin{eqnarray*}
h_{mk}(y) =  (H_k \ast \phi_m)(y)\in \mathcal{C}^{\infty }_b(\R ),\quad \phi_m(y) = m^{-1} \phi (m^{-1} y).
\end{eqnarray*}
If the equality (\ref{Induc}) is correct for any $k$, then
\begin{eqnarray}
E \left( f(S_{t}) \prod_{k=1}^{d} h'_{mk}(S_{s}^k)  \right) =E \left( f(S_{t}) \Gamma_{s,t} \prod_{k=1}^{d}
\frac{h_{mk}(S_{s}^k)}{S_{s}^k}  \right). \label{Induc2}
\end{eqnarray}
On the one hand, $h_{mk}(y)$ converges to $H_{k}(y)$ except at $y=0$ and the absolute continuity of the
law of $S_s^k$ ensures that $h_{mk}(S_{s}^k)$ converges almost surely to $H_{k}(S_{s}^k)$.
Using the dominated convergence theorem, we prove the convergence of $h_{mk}(S_{s}^k)$ to $H_{k}(S_{s}^k)$ in $L^{p}(\Omega )$ for $p \geq 1$. By Cauchy-Schwarz inequality, we prove the convergence
\begin{eqnarray*}
E \left( f(S_{t}) \Gamma_{s,t} \prod_{k=1}^{d} \frac{h_{mk}(S_{s}^k)}{S_{s}^k} \right) \longrightarrow E \left( f(S_{t}) \Gamma_{s,t} \prod_{k=1}^{d} \frac{H_{k}(S_{s}^k)}{S_{s}^k}  \right).
\end{eqnarray*}
On the other hand, $h'_{mk}(y_k)= \int_{\R} H_{k}(z_k) \phi'_m(y_k-z_k) dz_k = \phi_m(y_k-x_k)$. Moreover, we observe that, according to our assumption, the distribution of the vector $(S_{s}^1,...,S_{s}^d,S_{t}^1,...,S_{t}^d)$
admits a density with respect to the Lebesgue mesure on $\R^d \times \R^d$ we denote it by $p(y,z)$
with $y=(y_1,...,y_d)$ and $z=(z_1,...,z_d)$, thus
\begin{eqnarray*}
E \left( f(S_{t}) \prod_{k=1}^{d} h'_{mk}(S_{s}^k)  \right) =
\int_{\R^d} f(z) \left( \int_{\R^d} \prod_{k=1}^d \phi_m(y_k-x_k) p(y,z) dy_1 ...dy_d \right) dz_1 ...dz_d
\end{eqnarray*}
Because $\int_{\R^d} \prod_{k=1}^d \phi_m(y_k-x_k) p(y,z) dy_1 ...dy_d $ converges to $ p(x,z) $, we have
\begin{eqnarray*}
E \left( f(S_{t}) \prod_{k=1}^{d} h'_{mk}(S_{s}^k)  \right) \longrightarrow
E \left( f(S_{t}) \prod_{k=1}^{d} \varepsilon_{x_k}(S_{s}^k )  \right),
\end{eqnarray*}
which concludes the first part of this proof.\\

To prove the induction (\ref{Induc}), we introduce the following notations:
\begin{eqnarray*}
\widehat{h}_{k}^d (x)= \prod^{d}_{i=k} \frac{h_i(x_i)}{x_i},\quad
\widehat{h'}_{k} (x)= \prod^{k}_{i=1} h'_i(x_i),\quad x=(x_1,...,x_d).
\end{eqnarray*}
The case $k=d$ is given by
\begin{eqnarray*}
\begin{array}{ccc}
E \left( f(S_{t}) \widehat{h'}_{d}(S_{s})  \right) & = & E \left( \frac{1}{s}
\int^s_0 f(S_{t}) \widehat{h'}_{d-1} (S_{s}) \frac{D_u^d h_d(S_{s}^d)}{D_u^dS_{s}^d} du \right)  \\ & = &
E \left( \frac{1}{s} \int^s_0 f(S_{t}) \widehat{h'}_{d-1} (S_{s})
\frac{D_u^d h_d(S_{s}^d)}{\sigma_{dd} (u) S_{s}^d} du \right),
\end{array}
\end{eqnarray*}
where we replaced $h'_d(S_{s}^d)$ by $\frac{D_u^d h_d(S_{s}^d)}{D_u^dS_{s}^d}$ in the first equality and
$D_u^dS_{s}^d$ by its value $\sigma_{dd} (u) S_{s}^d$ in the second equality. Using Lemma \ref{lem1} with \vspace{-3mm}
\begin{eqnarray*}
F = \frac{f (S_{t})}{S_{s}^d} \prod_{i=1}^{d-1} h'_i(S_{s}^i) = \frac{f (S_{t})}{S_{s}^d} \widehat{h'}_{d-1} (S_{s})
\end{eqnarray*}
and the fact that $ \widehat{h'}_{d-1} (S_{s}) $ does not depend on the $d^{th}$ coordinate of the Brownian motion yields
\begin{eqnarray}
\begin{array}{c}
E \left( \frac{1}{s} \int^s_0 f(S_{t}) \widehat{h'}_{d-1} (S_{s})
\frac{D_u^d h_d(S_{s}^d) du}{\sigma_{dd} (u) S^d_{s}} \right) \hspace{65mm}\\
\begin{array}{ccc}
& = &
E \left( F h_d(S_{s}^d) \frac{1}{s} \int^s_0 \frac{dW^d_u}{\sigma_{dd}(u)} \right)
- E \left( h_d(S_{s}^d) \frac{1}{s} \int^s_0 D_u^d \frac{ \widehat{h'}_{d-1} (S_{s}) f(S_{t})}{S_{s}^d} \frac{du}{\sigma_{dd}(u)} \right) \\ & = &
E \left( F h_d(S_{s}^d) \frac{1}{s} \int^s_0 \frac{dW^d_u}{\sigma_{dd}(u)} \right)-
E \left( \widehat{h'}_{d-1} (S_{s}) h_d(S_{s}^d) \frac{1}{s}
\int^s_0 D_u^d \frac{f(S_{t})}{S_{s}^d} \frac{du}{\sigma_{dd}(u)} \right).
\end{array}
\end{array}
\label{pr1}
\end{eqnarray}
Besides using Lemma \ref{lem2} for the Malliavin derivative of $f(S_t)$, we get for $v \in ]s,t[$
\begin{eqnarray*}
\frac{1}{\sigma_{dd} (u)} D_{u}^{d} \left[ \frac{f(S_{t})}{S_{s}^d} \right] =
\frac{1}{S_{s}^d \sigma_{dd} (v)} D_{v}^{d} f(S_{t}) - \frac{f(S_{t})}{S_{s}^d}.
\end{eqnarray*}
Thus, the value of the last term of (\ref{pr1}) is given by
\begin{eqnarray*}
\begin{array}{c}
E \left( \widehat{h'}_{d-1} (S_{s}) h_d(S_{s}^d) \frac{1}{s}
\int^s_0 D_u^d \frac{f(S_{t})}{S_{s}^d} \frac{du}{\sigma_{dd}(u)} \right)
= -E \left( \widehat{h'}_{d-1} (S_{s}) h_d(S_{s}^d) \frac{f(S_{t})}{S_{s}^d} \right) \hspace{7mm}
\\ + E \left( \widehat{h'}_{d-1} (S_{s}) \frac{h_d(S_{s}^d)}{S_{s}^d } \frac{1}{t-s}
\int^t_s D_{v}^{d} f(S_{t}) \frac{dv}{\sigma_{dd}(v)} \right).
\end{array}
\end{eqnarray*}
And by duality (\ref{dual1}) we remove the Malliavin derivative of $f(S_{t})$ in the last term
of the previous equality
\begin{eqnarray*}
\begin{array}{ccc}
E \left( \frac{\widehat{h'}_{d-1} (S_{s}) h_d(S_{s}^d)}{S_{s}^d } \frac{1}{t-s}
\int^t_s \frac{D_{v}^{d} f(S_{t}) dv}{\sigma_{dd}(v)} \right) & = &
E \left( \frac{\widehat{h'}_{d-1} (S_{s}) h_d(S_{s}^d)}{S_{s}^d } E \left\{ \frac{1}{t-s}
\int^t_s \frac{D_{v}^{d} f(S_{t}) dv}{\sigma_{dd}(v)} \Big| \mathcal{F}_{s} \right\}\right)
\\ & = & E \left( \frac{\widehat{h'}_{d-1} (S_{s}) h_d(S_{s}^d)}{S_{s}^d } E \left\{
f(S_{t}) \frac{1}{t-s} \int_{s}^{t} \frac{dW^{d}_v}{\sigma_{dd}(v)} \Big| \mathcal{F}_{s} \right\} \right).
\end{array}
\end{eqnarray*}
Regrouping all terms together
\begin{eqnarray*}
E \left( f(S_{t}) \widehat{h'}_{d} (S_{s}) \right) =E \left( f(S_{t}) \Gamma^{d}_{s,t} \widehat{h'}_{d-1} (S_{s}) \widehat{h}_{d}^d (S_{s}) \right) ,\quad  \Gamma^{d}_{s,t} = \pi^{d,d}_{s,t}.
\end{eqnarray*}

Let us suppose that (\ref{Induc}) is satisfied for $k$ and prove it for $k-1$, thus
\begin{eqnarray*}
\begin{array}{ccc}
E \left( f(S_{t}) \widehat{h'}_{d-1} (S_{s}) \right) & = & E \left( f(S_{t}) \Gamma^{k+1}_{s,t}
\widehat{h}_{k+1}^d (S_{s}) \widehat{h'}_{k} (S_{s}) \right) \\ & = &
E \left( \frac{1}{s} \int^s_0  f(S_{t}) \Gamma^{k+1}_{s,t} \widehat{h}_{k+1}^d (S_{s})
\widehat{h'}_{k-1} (S_{s}) \frac{D_u^k h_k(S_{s}^k)}{\sigma_{kk}(u) S_{s}^k} du \right)\\ & = &
E \left( \frac{1}{s} \int^s_0 \frac{ f(S_{t}) \Gamma^{k+1}_{s,t} \widehat{h}_{k+1}^d (S_{s})
\widehat{h'}_{k-1} (S_{s})}{S_{s}^k} \frac{D_u^k h_k(S_{s}^k)}{\sigma_{kk}(u)} du \right),
\end{array}
\end{eqnarray*}
where we replaced $h'_k(S_{s}^k)$ by $\frac{D_u^k h_k(S_{s}^k)}{D_u^kS_{s}^k}$ in the second equality.
Using Lemma \ref{lem1} with \vspace{-3mm}
\begin{eqnarray*}
F = \frac{ f(S_{t}) \Gamma^{k+1}_{s,t} \widehat{h}_{k+1}^d (S_{s})
\widehat{h'}_{k-1} (S_{s})}{S_{s}^k}
\end{eqnarray*}
and the fact that $ \widehat{h'}_{k-1} (S_{s}) $ does not depend on the $j^{th}$ coordinate ($j \geq k$) of the Brownian motion yields
\begin{eqnarray}
\begin{array}{c}
E \left( \frac{1}{s} \int^s_0  \frac{F D_u^k h_k(S_{s}^k)}{\sigma_{kk}(u)} du \right) =
\sum_{j=k}^d E \left( F h_k(S_{s}^k) \frac{1}{s} \int^s_0 \rho_{jk}(u) dW^j_u \right) \hspace{20mm}
\\ \\ - \sum_{j=k}^d E \left( h_k(S_{s}^k) \widehat{h'}_{k-1} (S_{s}) \frac{1}{s} \int^s_0 D_u^j \left[
\frac{f (S_{t}) \widehat{h}_{k+1}^d (S_{s}) \Gamma^{k+1}_{s,t}}{S_{s}^k}  \right] \rho_{jk}(u) du \right).
\end{array}
\label{pr2}
\end{eqnarray}
Besides, if for $x=(x_1,...,x_d)$ we denote $\Pi(x) = \frac{ \widehat{h}_{k+1}^d (x)}{x_k}$, the Malliavin derivative
of the last term of (\ref{pr2}) provides
\begin{eqnarray*}
\begin{array}{ccc}
D_u^j \left[ \Gamma^{k+1}_{s,t} \Pi(S_{s}) f (S_{t}) \right] = D_u^j \Gamma^{k+1}_{s,t} \Pi(S_{s}) f (S_{t}) & + &
\Gamma^{k+1}_{s,t} D_u^j \Pi(S_{s}) f (S_{t}) \\ & + & \Gamma^{k+1}_{s,t} \Pi(S_{s}) D_u^j f (S_{t}).
\end{array}
\end{eqnarray*}
Using Lemma \ref{lem2} for the Malliavin derivative in the two last terms, we get
\begin{eqnarray}
\sum_{j=k}^d \rho_{jk}(u) D_u^j \Pi(S_s) = S^k_s \partial_{x_k} \Pi(S_{s}) = -\Pi(S_{s}),
\label{diffH}
\end{eqnarray}
\begin{eqnarray}
\sum_{j=k}^d \rho_{jk}(u) D_u^j f (S_{t}) = S^k_t \partial_{x_k} f (S_{t}).
\label{diffF}
\end{eqnarray}
From (\ref{diffH}), we deduce that
\begin{eqnarray*}
\widehat{h'}_{k-1} (S_{s}) h_{k} (S^k_{s}) f (S_{t}) \Gamma^{k+1}_{s,t}
\frac{1}{s} \int^s_0  \sum_{j=k}^d  \rho_{jk}(u) D_u^j \Pi(S_s) du = -
\frac{\widehat{h'}_{k-1} (S_{s}) \widehat{h}_{k}^d (S_{s}) f (S_{t})
\Gamma^{k+1}_{s,t}}{S^k_{s}}.
\end{eqnarray*}
Thus, introducing the random variable $ \widetilde{\pi}_{s,t}^{k,d}  = 1 + \sum_{j=k}^d \int^s_0 \rho_{jk}(u) dW^j_u$ and using (\ref{pr2})
\begin{eqnarray}
\begin{array}{ccc}
E \left( \frac{1}{s} \int^s_0  \frac{F D_u^k h_k(S_{s}^k)}{\sigma_{kk}(u)} du \right) & = &
E \left( \F h_k(S_{s}^k) \widetilde{\pi}_{s,t}^{k,d} \right)
\\ & - & E \left( \frac{\widehat{h}_{k}^d (S_{s}) \widehat{h'}_{k-1} (S_{s}) f (S_{t})}
{S_{s}^k} \frac{1}{s} \int^s_0 \sum_{j=k}^d \rho_{jk}(u) D_u^j \Gamma^{k+1}_{s,t} du \right) \\ & - &
E \left( \frac{\widehat{h}_{k}^d (S_{s}) \widehat{h'}_{k-1} (S_{s}) \Gamma^{k+1}_{s,t}}
{S_{s}^k} \frac{1}{t-s} \int^t_s \sum_{j=k}^d \rho_{jk}(u) D_u^j f (S_{t}) du \right),
\end{array}
\label{pr3}
\end{eqnarray}
where we used the fact (\ref{diffF}) that $\sum_{j=k}^d \rho_{jk}(u) D_u^j f (S_{t})$ does not depend on $u$.
Let us develop the last term of (\ref{pr3})
\begin{eqnarray*}
\begin{array}{c}
E \left( \frac{\widehat{h}_{k}^d (S_{s}) \widehat{h'}_{k-1} (S_{s}) \Gamma^{k+1}_{s,t}}
{S_{s}^k } \frac{1}{t-s} \int^t_s \sum_{j=k}^d \rho_{jk}(u) D_u^j f (S_{t}) du \right) \hspace{40mm}\\
\begin{array}{cc}
= & E \left( \frac{\widehat{h}_{k}^d (S_{s}) \widehat{h'}_{k-1} (S_{s})}
{S_{s}^k } \sum_{j=k}^d E \left[ \frac{1}{t-s} \int^t_s \Gamma^{k+1}_{s,t} \rho_{jk}(u) D_u^j f (S_{t}) du \Big|
\mathcal{F}_{s} \right] \right) \\ = & E \left( \frac{\widehat{h}_{k}^d (S_{s}) \widehat{h'}_{k-1} (S_{s})}
{S_{s}^k } \sum_{j=k}^d E \left[f (S_{t}) \frac{1}{t-s} \int^t_s \Gamma^{k+1}_{s,t} \rho_{jk}(u)
\delta W_u^j \Big| \mathcal{F}_{s} \right] \right) \\ = & \sum_{j=k}^d
E \left( F h_k(S_{s}^k) \frac{1}{t-s} \int^t_s \rho_{jk}(u) d W_u^j \right) \\ - & \sum_{j=k}^d
E \left( \frac{f (S_{t}) \widehat{h}_{k}^d (S_{s}) \widehat{h'}_{k-1} (S_{s})}
{S_{s}^k } \frac{1}{t-s} \int^t_s \rho_{jk}(u) D_u^j \Gamma^{k+1}_{s,t} du \right).
\end{array}
\end{array}
\label{pr4}
\end{eqnarray*}
We applied (\ref{dual1}) in the third equality to remove the Malliavin derivative of $f(S_{t})$.
We also used (\ref{IP1}) in the last equality. To complete the proof, we should remark that
\begin{eqnarray*}
\frac{1}{s} \int^s_0 D_u^j \Gamma^{k+1}_{s,t} \rho_{jk}(u) du
- \frac{1}{t-s} \int^t_s D_v^j \Gamma^{k+1}_{s,t} \rho_{jk}(v) dv  = - \int_{0}^{t} D_y^j \Gamma^{k+1}_{s,t} D_y^j \pi^{k,d}_{s,t} dy
\end{eqnarray*}
and because $\Gamma^{k+1}_{s,t}$ is an $\mathcal{F}_{t}^{k+1,...,d}$-measurable random variable
\begin{eqnarray*}
\Gamma_{s,t}^k = \Gamma_{s,t}^{k+1} \pi_{s,t}^{k,d} - \sum_{j=k}^d \int_{0}^{t} D_u^j \Gamma^{k+1}_{s,t} D_u^j \pi^{k,d}_{s,t} du = \Gamma_{s,t}^{k+1} \pi_{s,t}^{k,d} - \sum_{j=k+1}^d \int_{0}^{t} D_u^j \Gamma^{k+1}_{s,t} D_u^j \pi^{k,d}_{s,t} du.
\end{eqnarray*}
}
Theorem \ref{theoExp} provides the expression of $\Gamma_{s,t}^{k}$
in (\ref{gammaexp}) without using the Malliavin derivatives $\{ D_u^j \}_{j>k}$
and which can be efficiently computed using (\ref{recdel}).
We will use in Theorem \ref{theoExp} the set of the second order permutations
$\overline{\mathcal{S}}_{k,d}$ defined as the following
\begin{eqnarray}
\overline{\mathcal{S}}_{k,d} = \{ p \in \mathcal{S}_{k,d}, \hspace{1mm} p \circ p = Id \},
\label{symdef}
\end{eqnarray}
where $\mathcal{S}_{k,d}$ is the set of permutations on $\{k,...,d\}$ and $Id$ is the
identity application. By induction, one can easily prove that
\begin{eqnarray}
\hspace{10mm} \overline{\mathcal{S}}_{k,d} = \{ \tau^k_k \circ p , \hspace{1mm} p \in \overline{\mathcal{S}}_{k+1,d} \}
\cup \{ \tau^l_k \circ p , \hspace{1mm} p \in \overline{\mathcal{S}}_{k+1,d}, \hspace{1mm} p(l)=l,
\hspace{1mm} l \in \{k+1,...,d \}\},
\label{recper}
\end{eqnarray}
with $\tau_i^j:i \mapsto j$ as the transition application on $\{k,...,d\}$. We also denote by
$\Delta $ the determinant that involves only the permutations of $\overline{\mathcal{S}}_{k,d}$,
that is to say, the $\Delta $ associated to the matrix $C = \{ C_{i,j} \}_{k \leq i,j \leq d}$ is given by
\begin{eqnarray*}
\Delta = \sum_{p \in \overline{\mathcal{S}}_{k,d}} \epsilon (p) \prod_{i=1}^d C_{i,p (i)}
\end{eqnarray*}
Using (\ref{recper}), we can easily prove that
\begin{eqnarray}
\Delta = C_{k,k} \Delta_{k,k} + \sum_{i=k+1}^d \epsilon (\tau^i_k) C_{i,k} C_{k,i} \Delta_{k,i}
\label{recdel}
\end{eqnarray}
where $\Delta_{k,i}$ is the $\Delta $ associated to the $C^{i,k}$ obtained from $C$ by suppressing the
line and the column $i$ as well as the line and the column $k$. Based on the development
according to the first line, relation (\ref{recdel}) provides a recursive formula even more efficient than the determinant formula.
Of course, we can generalize the relation (\ref{recdel}) to the one that
involves the development according to a $j^{th}$ line or a $j^{th}$ column
with $k \leq j \leq d$.

\begin{theorem}
Based on the assumptions and the results of Theorem \ref{theo1},
for $k \in \{ 1,...,d \}$ the value of $\Gamma_{s,t}^k$ is given by
\begin{eqnarray}
\Gamma_{s,t}^k = \sum_{p \in \overline{\mathcal{S}}_{k,d}} \epsilon (p) A_{k,p (k)}A_{k+1,p (k+1)}...A_{d,p (d)}
= \sum_{p \in \overline{\mathcal{S}}_{k,d}} \epsilon (p) \prod_{i=k}^d A_{i,p (i)},
\label{gammaexp}
\end{eqnarray}
with $\epsilon (p)$ as the signature of the permutation $p \in \overline{\mathcal{S}}_{k,d}$, $ \overline{\mathcal{S}}_{k,d}$ defined in
(\ref{symdef}) and
\begin{eqnarray*}
A =
\left( \begin{array}{ccccc}
\pi^{1,d}_{s,t} & C_{1,2} & C_{1,3} & \cdots & C_{1,d}\\
1 & \pi^{2,d}_{s,t} & C_{2,3} & \cdots & C_{2,d} \\
\vdots & \ddots & \ddots & \ddots & \vdots\\
1 &  \cdots &  1 &  \pi^{d-1,d}_{s,t} & C_{d-1,d} \\
1 & 1 & \cdots & 1 & \pi^{d,d}_{s,t}
\end{array}
\right) ,
\end{eqnarray*}
where $C_{k,l}$ is the covariance of $\pi^{k,d}_{s,t}$ and $\pi^{l,d}_{s,t}$.
\label{theoExp}
\end{theorem}

\begin{proof}
We prove (\ref{gammaexp}) by a decreasing induction. For $k=d$, the
expression (\ref{gammaexp}) is clearly satisfied. We suppose that (\ref{gammaexp})
is satisfied for $k+1$ and we prove it for $k$. According to Theorem \ref{theo1},
$\Gamma_{s,t}^k = \Gamma_{s,t}^{k+1} \pi_{s,t}^{k,d} - \sum_{j=k+1}^d \int_{0}^{t} D_u^j \Gamma^{k+1}_{s,t} D_u^j \pi^{k,d}_{s,t} du$,
but
\begin{eqnarray*}
\begin{array}{ccc}
D_u^j \Gamma_{s,t}^{k+1} & =  & \sum_{l=k+1}^d \sum_{p \in \overline{\mathcal{S}}_{k+1,d}} \epsilon (p)
\prod_{i=k+1,i \neq l}^d A_{i,p (i)} D_u^j A_{l,p (l)} \\ \\ & = & \sum_{l=k+1}^d
\sum_{p \in \overline{\mathcal{S}}_{k+1,d}, p(l)=l}
\epsilon (p) \prod_{i=k+1,i \neq l}^d A_{i,p (i)} D_u^j A_{l,l},
\end{array}
\end{eqnarray*}
the second equality is due to the fact that $A_{l,p (l)}$ is a constant except for $p(l)=l$. Subsequently
\begin{eqnarray*}
\begin{array}{c}
- \sum_{j=k+1}^d \int_{0}^{t} D_u^j \Gamma^{k+1}_{s,t} D_u^j \pi^{k,d}_{s,t} du \hspace{75mm}\\ \\
\begin{array}{cc}
= & - \sum_{l=k+1}^d \sum_{p \in \overline{\mathcal{S}}_{k+1,d}, p(l)=l}
\epsilon (p) \prod_{i=k+1,i \neq l}^d A_{i,p (i)} \sum_{j=k+1}^d \int_{0}^{t} D_u^j A_{l,l} D_u^j \pi^{k,d}_{s,t} \\ \\ = & - \sum_{l=k+1}^d \sum_{p \in \overline{\mathcal{S}}_{k+1,d}, p(l)=l} \epsilon (p) \prod_{i=k+1,i \neq l}^d A_{i,p (i)} C_{k,l}.
\end{array}
\end{array}
\end{eqnarray*}
Finally
\begin{eqnarray*}
\Gamma_{s,t}^k = \Gamma_{s,t}^{k+1} \pi_{s,t}^{k,d} - \sum_{j=k+1}^d \int_{0}^{t} D_u^j \Gamma^{k+1}_{s,t} D_u^j \pi^{k,d}_{s,t} du \hspace{54.5mm} \\
 = \pi_{s,t}^{k,d} \sum_{p \in \overline{\mathcal{S}}_{k+1,d}} \epsilon (p) \prod_{i=k+1}^d A_{i,p (i)}
- \sum_{l=k+1}^d C_{k,l} \sum_{p \in \overline{\mathcal{S}}_{k+1,d}, p(l)=l} \epsilon (p) \prod_{i=k+1,i \neq l}^d A_{i,p (i)} \\
= \sum_{p \in \overline{\mathcal{S}}_{k,d}} \epsilon (p) \prod_{i=k}^d A_{i,p (i)}. \hspace{84mm}
\end{eqnarray*}
The last equality is due to the development of $\sum_{p \in \overline{\mathcal{S}}_{k,d}} \epsilon (p) \prod_{i=k}^d A_{i,p (i)}$ according to the $k^{th}$ line of $A$ which can be justified by (\ref{recper}).
\end{proof}

As a corollary of Theorem \ref{theo1} and Theorem \ref{theoExp}, we obtain the following result.

\begin{corollary}
For any $s \in ]0,t[$, $g \in \mathcal{E}_{b}$ and $x>0$, if $\sigma_{ij}(t) = \sigma_{ij} \delta (i-j)$ then
\begin{eqnarray*}
E \left( g(S_{t}) \Big| S_{s} = x \right) =
\frac{T_{s,t}[g](x)}{T_{s,t}[1](x)},
\end{eqnarray*}
with
\begin{eqnarray}
T_{s,t}[f ](x)=E \left( f(S_{t})
\prod_{k=1}^{d}\frac{H_{k}(S_{s}^k)  W_{s,t}^{k} }
{\sigma_{k} s (t-s) S^k_{s}}
\right), \label{TTd}
\end{eqnarray}
and
\begin{eqnarray*}
 W_{s,t}^{k} =(t-s)(W_{s}^{k}+ \sigma_{k}
s)-s(W_{t}^{k}-W_{s}^{k}) ,\quad k=1,...,d.
\end{eqnarray*}
\label{theo5}
\end{corollary}

\section{Variance reduction method based on conditioning \label{sec3}}

In this section, we show that one can reduce the variance by a projection on $L^{2} \left( \left\{ \int_{0}^{t}
\sigma_{ij}(u) dW^j_u \right\}_{i, j }\right) $ and by using a closed formula of $T_{s,t}[1 ](x)$.
Like in section \ref{sec2}, we give in Theorem \ref{theo6} the results of the special case
$\sigma_{ij}(t) = \sigma_{ij} \delta (i-j)$ ($\sigma_{ij}$ is a constant)
that will be used to test our variance reduction method.\\

We begin with $T_{s,t}[1 ](x)$, we can compute the explicit value of this function of $x$.
The $T_{s,t}[1 ](x)$ closed formula can be got, for instance, from a change of probability. Indeed,
we define the probability $\mathbb{P} = N_{coeff} ( \prod_{k=1}^{d} S_{0}^k/S_{s}^k ) P$ which yields
\begin{eqnarray*}
T_{s,t}[1 ](x)= \frac{1}{N_{coeff}} \mathbb{E} \left( \left[ \prod_{k=1}^{d} H_{k}(S_{s}^k) \right] \Gamma_{s,t}  \right),
\end{eqnarray*}
$N_{coeff}$ is a deterministic normalization coefficient such that $M_s = N_{coeff} ( \prod_{k=1}^{d} S_{0}^k/S_{s}^k )$
is an exponential martingale with $E(M_s)=1$. Under $\mathbb{P}$, $\Gamma_{s,t} $ has the same law as a polynomial of Gaussian variables which is sufficient to conduct the computations.\\

Let us now denote
\begin{eqnarray*}
h(x,\{ y_{ij} \}_{j \leq i})=E \left( \Gamma_{s,t} \prod_{k=1}^{d}
\frac{H_{k}(S_{s}^k)}{S_{s}^k} \big| \left\{ \int_{0}^{t}
\sigma_{ij}(u) dW^j_u \right\}_{1 \leq j \leq i \leq d} = \{ y_{ij} \}_{1 \leq j \leq i \leq d} \right)
\end{eqnarray*}
In what follows, we are going to prove that the function $h(x,\{ y_{ij} \}_{1 \leq  j \leq i \leq d})$
can be explicitly known if, for each $j$, the $(d-k) \times (d-k)$ matrix
$\Sigma_{jt} = \left\{ \Sigma^{ik}_{jt} \right\}_{j \leq i,k \leq d} = \left\{ \int^t_0 \sigma_{ij} (u)
\sigma_{kj} (u) du \right\}_{j \leq i,k \leq d} $ is invertible. First, please remark that according to our notations
$i-j+1$ and $k-j+1$ are the indices of the element $\Sigma^{ik}_{jt}$ in the matrix $\Sigma_{jt}$ (we will use
similar convention also for $A^{j}$, $B^j$, $\Psi_{jt} $ and $\Phi_{jt} $). Also we notice that
the condition of invertibility of $\Sigma_{jt}$ is not an important constraint, because one can choose a time
discretization $\{ t_m \}$ such that the matrices $ \left\{ \Sigma_{j t_m} \right\}_{k \leq d}$ fulfill this
condition\footnote{Nevertheless, this is a difficult task when the dimension is sufficiently big.}.

The computation of $h(x,\{ y_{ij} \}_{1 \leq j \leq i \leq d})$ is based
on a regression of Gaussian variables according to the
Gaussian variables $Y_{ij} = \int_{0}^{t} \sigma_{ij} (u) dW^j_u $. First,
we perform a linear regression of $\int_{0}^{t} \varphi_{jk}(u) dW^{j}_u$ according to $Y_{ij}$
\begin{eqnarray}
\int_{0}^{t} \varphi_{jk}(u) dW^{j}_u = \sum^d_{i=j} a^j_{i,k} Y_{ij} + X_{jk},
\label{reg1}
\end{eqnarray}
with $\left\{ X_{jk} \right\}_{1 \leq k \leq  j \leq d}$ as a Gaussian vector $\mathcal{N}(0,C_X)$ which is orthogonal to $Y$.
Using Itô isometry twice and the orthogonality of $Y$ and $X$, we obtain
\begin{eqnarray*}
E \left( \int_{0}^{t} \varphi_{jk}(u) dW^{j}_u Y_{lj} \right) = \int^t_0 \varphi_{jk}(u) \sigma_{lj}(u) du
= \sum^n_{j=k} \Sigma^{li}_{jt} a^j_{i,k}.
\end{eqnarray*}
If we denote $A^{j} = \{ a^j_{i,k} \}_{j \leq i,k \leq d}$ and $\Psi_{jt} = \left\{ \int^t_0 \varphi_{jk}(u) \sigma_{lj}(u) du \right\}_{ k,l}$,
we get
\begin{eqnarray*}
A^{j} = \Sigma_{jt}^{-1} \Psi_{jt}.
\end{eqnarray*}
In the same way, we perform a linear regression of $\int_{0}^{s} \sigma_{kj}(u) dW^{j}_u$ according to $Y_{ij}$
\begin{eqnarray}
\int_{0}^{s} \sigma_{kj}(u) dW^{j}_u = \sum^d_{i=j} b^j_{i,k} Y_{ij} + Z_{kj},
\label{reg2}
\end{eqnarray}
with $\left\{ Z_{kj} \right\}_{1 \leq j \leq  k \leq d}$ as a Gaussian vector $\mathcal{N}(0,C_Z)$ which is orthogonal to $Y$.
Using Itô isometry twice and the orthogonality of $Y$ and $Z$, we obtain
\begin{eqnarray*}
E \left( \int_{0}^{s} \sigma_{kj}(u) dW^{j}_u Y_{lj} \right) = \int^s_0 \sigma_{kj}(u) \sigma_{lj}(u) du
= \sum^d_{i=j} \Sigma^{li}_{jt} b^j_{i,k}.
\end{eqnarray*}
If we denote $B^{j} = \{ b^j_{i,k} \}_{j \leq i,k \leq d}$, we get
\begin{eqnarray*}
B^{j} = \Sigma_{jt}^{-1} \Sigma_{js}.
\end{eqnarray*}
Now using (\ref{reg1}), (\ref{reg2}) and the value of $A$ and $B$, the covariance matrices $C_X$, $C_Z$ and $C_{XZ}=E(XZ)$ are given by ($\Phi_{jt}^{i,k} = \int^t_0 \varphi_{ji}(u) \varphi_{jk}(u) du$)
\begin{eqnarray*}
[C_X]^j_{i,k} = E(X_{ji} X_{jk}) = \Phi_{jt}^{i,k} - (A_{k}^j)' \Psi_{jt}^i - (A_{i}^j)' \Psi_{jt}^k + (A_{k}^j)' \Sigma_{jt} A_{i}^j,
\end{eqnarray*}
\begin{eqnarray*}
[C_Z]^j_{i,k} = E(Z_{ij} Z_{kj}) = \Sigma_{js}^{i,k} - (B_{k}^j)' \Sigma_{js}^{i} - (B_{i}^j)' \Sigma_{js}^{k} + (B_{k}^j)' \Sigma_{jt} B_{i}^j,
\end{eqnarray*}
\begin{eqnarray*}
[C_{XZ}]^j_{i,k} = E(X_{ji} Z_{kj}) = \Psi_{js}^{i,k} - (A_{k}^j)' \Sigma_{js}^{i} - (B_{i}^j)' \Psi_{jt}^k + (A_{k}^j)' \Sigma_{jt} B_{i}^j.
\end{eqnarray*}
Using (\ref{reg1}) and (\ref{reg2}), we express $\Gamma_{s,t}$ and $S^k_s$ according to $Y_{ij}$, $Z_{ij}$ and $X_{ji}$ then we conduct standard Gaussian computations to obtain the expression of $h(x, y_{ij} )$~\footnote{One can use Mathematica to compute it formally.}. In Theorem \ref{theo6}, we give an explicit expression of $T_{s,t}[1 ](x)$ and
$h(x, y_{ij} )$ in the case of multi-dimensional B\&S models with independent coordinates.\\

We can see that now that we know the explicit value of $T_{s,t}[1 ](x)$ and
$h(x,\{ y_{ij} \}_{1 \leq  j \leq i \leq d})$, subsequently, we should choose between the simulation of:
\begin{itemize}
  \item[\textbf{P1)}] $N$ paths of $g(S_t) h\left( x,\{ \int_{0}^{t}
\sigma_{ij}(u) dW^j_u \}_{i,j} \right)$ then set the continuation to the value
\begin{eqnarray*}
C(x) := \frac{\frac{1}{N} \sum_{l=1}^N g^{l}(S_t) h\left( x,\{ \int_{0}^{t}
\sigma_{ij}(u) dW^j_u \}^l_{1 \leq j \leq i \leq d} \right)}{T_{s,t}[1 ](x)}.
\end{eqnarray*}
  \item[\textbf{P2)}] $N'$ paths of $g(S_t) h\left( x,\{ \int_{0}^{t}
\sigma_{ij}(u) dW^j_u \}_{i,j} \right) $ and $N$ paths of
$ h\left( x,\{ \int_{0}^{t} \sigma_{ij}(u) dW^j_u \}_{i,j} \right)$
then set the continuation to the value
\begin{eqnarray*}
C(x) := \frac{\frac{1}{N'} \sum_{l=1}^{N'} g^{l}(S_t) h\left( x,\{ \int_{0}^{t}
\sigma_{ij}(u) dW^j_u \}^l_{1 \leq j \leq i \leq d} \right)}
{\frac{1}{N} \sum_{l=1}^N h\left( x,\{ \int_{0}^{t}
\sigma_{ij}(u) dW^j_u \}^l_{1 \leq j \leq i \leq d} \right)}.
\end{eqnarray*}
\end{itemize}
Based on a variance reduction argument, Theorem \ref{theo4} will indicate the preferable
method to use.

\begin{theorem}
For any $s \in ]0,t[$, $g \in \mathcal{E}_{b}$ and $x>0$, if $\sigma_{ij}(t) = \sigma_{ij} \delta (i-j)$ then the
conditional expectation given in Theorem \ref{theo5} can be reduced to
\begin{eqnarray*}
E \left( g(S_{t}) \Big| S_{s} = x \right) = \frac{ E \left(
g(S_{t}) \prod_{k=1}^{d} \sqrt{t} \exp \left( \frac{-s
\sigma_{k}}{t} \left( \frac{s \sigma_{k}}{t} + W_{t}^{k} \right)
- \frac{(d_{2k}(W_{t}^{k}) + m_k)^{2}}{2} \right) \right)
}{\prod_{k=1}^{d}\sqrt{(t-s)} e^{-\frac{d_{1k}^{2}}{2}}},
\end{eqnarray*}
with
\begin{eqnarray*}
m_k=\sigma_{k}\sqrt{\frac{s(t-s)}{t}},\quad d_{2k}(W_{t}^{k})=\sqrt{\frac{t}{s(t-s)}} \left( \beta_k -
\frac{s W_{t}^{k}}{t}\right) ,\quad  d_{1k}=\frac{\beta_k+ \sigma_{k} s}{\sqrt{s}},
\end{eqnarray*}
where
\begin{eqnarray*}
\beta_k = \frac{1}{\sigma_{k}} \left( \ln \left[ \frac{x_{k}}{S_{0}^{k}} \right] +
\frac{\sigma_{k}^2}{2} \right).
\end{eqnarray*}
\label{theo6}
\end{theorem}

\begin{proof}
We simplify the constant $\sigma_{k} s
(t-s)$ in (\ref{TTd}) from the denominator and the numerator of
the conditional expectation $E \left( g(S_{t}) \Big| S_{s} = x \right)$,
then we use the independence of the coordinates to obtain
\begin{eqnarray*}
E \left( g(S_{t}) \Big| S_{s} = x \right) = \frac{ E
\left(g(S_{t}) \prod_{k=1}^{d}
H_{i}(S_{s}^k)  W_{s,t}^{k} /
S_{s}^{k} \right) }{\prod_{k=1}^{d}E \left(
H_{k}(S_{s}^k)  W_{s,t}^{k} /
S_{s}^{k} \right)}.
\end{eqnarray*}
Afterwards, we use the independence of the increments to obtain
\begin{eqnarray*}
\begin{array}{c}
E\left( \frac{H_{k}(S_{s}^k)} {
S_{s}^{k}} W_{s,t}^{k}\right) = E\left(
\frac{H_{k}(S_{s}^k)} {
S_{s}^{k}} [(t-s)(W_{s}^{k}+\sigma_k
s)-s(W_{t}^{k}-W_{s}^{k})] \right) \\ = (t-s) E\left(
\frac{H_{k}(S_{s}^k)}
{S_{s}^{k}} (W_{s}^{k}+\sigma_k s)\right) - s E\left(
\frac{H_{k}(S_{s}^k)} {
S_{s}^{k}}\right) E(W_{t}^{k}-W_{s}^{k})
 \\ = (t-s) E \left(
\frac{H_{k}(S_{s}^k)} {
S_{s}^{k}} (\sqrt{s} G +\sigma_{k} s)\right),
\end{array}
\end{eqnarray*}
where the random variable $G$ has a standard normal distribution.
Moreover we have the following equality in distribution
\begin{eqnarray*}
S_{s}^{k} \doteq S_{0}^{k} \exp \left( -\frac{\sigma_{k}^2}{2}s +
\sigma_{k} \sqrt{s} G\right).
\end{eqnarray*}
Computing the expectation we obtain
\begin{eqnarray}
E \left( \frac{H_{k}(S_{s}^k)} {
S_{s}^{k}} (\sqrt{s} G +\sigma_{k} s)\right) =
\alpha_{k} (t-s) \sqrt{\frac{s}{2 \pi }} e^{-\frac{d_{1k}^{2}}{2}},
\label{calden}
\end{eqnarray}
with $\alpha_{k} = e^{ \sigma^{2}_{k} s}$.

Regarding the numerator, we condition according to $W^{k}_{t}=w^{k}$ and we
use the independence of coordinates
\begin{eqnarray*}
E \left(g(S_{t}) \prod_{k=1}^{d}
H_{k}(S_{s}^k)  W_{s,t}^{k} /
S_{s}^{k} \right) = E \left(g(S_{t}) \prod_{k=1}^{d}
h_{k}(W^{k}_{t}) \right),
\end{eqnarray*}
with
\begin{eqnarray}
h_{k}(w^{k}) = E \left( H_{k}(S_{s}^k)  W_{s,t}^{k} / S_{s}^{k} \Big| W^{k}_{t}=w^{k}
\right). \label{fcond}
\end{eqnarray}
Knowing $W^{k}_{0}=0$ and $W^{k}_{t}=w^{k}$, when we fix
$s$ the random variable $W^{k}_{s} \doteq \frac{s w^{k}}{t} +
\sqrt{\frac{s(t-s)}{t}} G $ and $G$ has a standard normal
distribution. Also, we have the following equality in distribution
for $ W_{s,t}^{k}$: $ W_{s,t}^{k} \doteq (t-s)
\sigma_{k} s + \sqrt{ts(t-s)} G$ and $S_{s}^{k}
\doteq S_{0}^{k} \exp \left( -\frac{\sigma_{k}^2}{2}s + \sigma_{k}
\frac{s w^{k}}{t} + \sigma_{k} \sqrt{\frac{s(t-s)}{t}} G
\right)$. Then we compute (\ref{fcond}) which yields:
\begin{eqnarray}
h_{k}(w^{k}) \doteq \alpha_{k} \sqrt{\frac{ts(t-s)}{2 \pi }} \exp
\left( \frac{-s \sigma_{k}}{t} \left( \frac{s \sigma_{k}}{t} +
w^{k} \right) - \frac{(d_{2k}(w^{k}) + m_k)^{2}}{2} \right),
\label{calfcond}
\end{eqnarray}
with $\alpha_{k} = e^{ \sigma^{2}_{k} s}$.

Using (\ref{calden}) and (\ref{calfcond}) we obtain the
requested result.
\end{proof}

\section{Advanced variance reduction method\label{sec4}}

We present, in this section, a less intuitive idea of variance
reduction that is based on an appropriate relation between
$N$ and $N'$ in (\ref{condapp}). This method can be applied
independently from conditioning detailed in previous section.

\begin{lemma}
Let $(X_{k})_{k \in \N^*}$ be a sequence of independent $\R^n$-valued
random variables that have the same law. We suppose that $X_{k}$ is square
integrable and we denote $\mu = E (X_{k})$, $C_{i,j} = Cov(X_i, X_j)$.
Let $r > 0$, $V_{\mu } = \{x \in \R^n, ||x-\mu ||_{\R^{n}} < r \}$ and
$g:\R^{n} \rightarrow \R$ such that $ g \in \mathcal{C}^1 ( V_{\mu })$, then
we have the following limits when $N \rightarrow \infty $
\begin{eqnarray*}
g(\overline{X}_N) \longrightarrow g( \mu) \hspace{2mm}a.s. ,\quad \sqrt{N} (g(\overline{X}_N) - g(\mu )) \longrightarrow \mathcal{N} (0,\Sigma) \hspace{2mm} in \hspace{1mm} law,
\end{eqnarray*}
such that
\begin{eqnarray}
\overline{X}_N= \frac{1}{N} \sum^{N}_{i=1} X_i ,\quad \Sigma = \left( \frac{\partial g}{\partial x_1},..., \frac{\partial g}{\partial x_n} \right)_{x=\mu } C \left( \frac{\partial g}{\partial x_1},..., \frac{\partial g}{\partial x_n} \right)_{x=\mu}^t.
\label{lemclef}
\end{eqnarray}
\label{lem3}
\end{lemma}

\begin{proof}
The almost sure convergence of $g(\overline{X}_N)$ results from the law of the large numbers and from
the continuity of $g$ in $\mu $. For the same reasons, the gradient vector $\frac{\partial g}{\partial x} (\overline{X}_N)$ converges a.s. to
$\frac{\partial g}{\partial x} (\mu)$. Besides $ \sqrt{N} (g(\overline{X}_N) - g(\mu )) =  \frac{\partial g}{\partial x} (\overline{X}_N) \cdot
\sqrt{N}(\overline{X}_N-\mu ) + \sqrt{N} (\overline{X}_N-\mu ) \cdot \epsilon (\overline{X}_N-\mu )$ and using the Slutsky Theorem, with
$G \sim \mathcal{N} (0, C)$
\begin{itemize}
\item $(\frac{\partial g}{\partial x} (\overline{X}_N),\sqrt{N}(\overline{X}_N-\mu))$ converges in law to
$(\frac{\partial g}{\partial x} (\mu),G)$.
\item $(\epsilon (\overline{X}_N-\mu ), \sqrt{N}(\overline{X}_N-\mu))$ converges in law to $(0,G)$.
\end{itemize}
Finally, because $(x,y) \mapsto xy$ and $(x,y) \mapsto x+y$ are continuous, then $ \sqrt{N} (g(\overline{X}_N) - g(\mu ))$ converges in law to $\frac{\partial g}{\partial x} (\mu)G$.
\end{proof}

Let us denote $Q$ as the quotient given by
\begin{eqnarray}
Q = \frac{ \frac{1}{N'} \sum_{i=1}^{N'}
X_i}{ \frac{1}{N} \sum_{i=1}^{N} Y_i}
\label{quot}
\end{eqnarray}
If $| E(Y_i) | \geq \varepsilon > 0$, according to Lemma \ref{lem3} $Q$ converges to $E(X_i)/E(Y_i)$.
In the following two theorems we will prove that we can accelerate the speed of convergence when acting
on the relation between $N$ and $N'$. We analyze the two cases:
\begin{itemize}
\item[\textbf{case 1:}] $N'=\lambda_1 N$ with $\lambda_1 \in [1/N,1]$ and we normalize (\ref{quot})
\begin{eqnarray*}
Q = \frac{ \frac{1}{N'} \sum_{i=1}^{N'}
X_i}{ \frac{1}{N} \left( \frac{N'}{N'} \sum_{i=1}^{N'} Y_i + \frac{N-N'}{N-N'}
\sum_{i=1}^{N-N'} Y_i \right) } = \frac{A_{N'}}{\lambda_1 B_{N'} + (1-\lambda_1) B_{N,N'}},
\label{quot1}
\end{eqnarray*}
where
\begin{eqnarray*}
A_{N'} = \frac{1}{N'} \sum_{i=1}^{N'}
X_i,\quad  B_{N'} =  \frac{1}{N'} \sum_{i=1}^{N'} Y_i,\quad  B_{N,N'} = \frac{1}{N-N'}
\sum_{i=1}^{N-N'} Y_i.
\end{eqnarray*}
We set $g_1(x,y,z) = x/(\lambda_1 y + (1-\lambda_1) z)$ and (\ref{lemclef}) provides
\begin{eqnarray}
\Sigma_1 (\lambda_1 ) = \frac{1}{B^2} \left( (2 \lambda_1^2 -  2 \lambda_1 +1) \frac{A^2}{B^2} \sigma^2_2 +
\sigma_{1}^2 - \frac{2 \lambda_1 A}{B} \sigma_1 \sigma_2 \rho \right),
\label{Sig1}
\end{eqnarray}
with $A=E(X)$, $B=E(Y)$, $\sigma_1^2 = Var(X)$, $\sigma_2^2 = Var(Y)$ and $\rho = Cov(X,Y)/(\sigma_1 \sigma_2)$.

\item[\textbf{case 2:}] $N=\lambda_2 N'$ with $\lambda_2 \in [1/N',1]$ and we normalize (\ref{quot})
\begin{eqnarray*}
Q = \frac{ \frac{1}{N'} \left( \frac{N}{N} \sum_{i=1}^{N} X_i + \frac{N'-N}{N'-N}
\sum_{i=1}^{N'-N} X_i \right)}{\frac{1}{N} \sum_{i=1}^{N}
Y_i} = \frac{\lambda_2 A_{N} + (1-\lambda_2) A_{N',N}}{B_{N}},
\label{quot2}
\end{eqnarray*}
where
\begin{eqnarray*}
A_{N} = \frac{1}{N} \sum_{i=1}^{N} X_i,\quad  A_{N',N} = \frac{1}{N'-N}
\sum_{i=1}^{N'-N} X_i,\quad  B_{N} =  \frac{1}{N} \sum_{i=1}^{N} Y_i.
\end{eqnarray*}
We set $g_2(x,y,z) = (\lambda_2 x + (1-\lambda_2) y)/z$ and (\ref{lemclef}) provides
\begin{eqnarray}
\Sigma_2 (\lambda_2 ) = \frac{1}{B^2} \left( (2 \lambda_2^2 -  2 \lambda_2 +1) \sigma^2_1 +
\frac{A^2}{B^2} \sigma^2_2 - \frac{2 \lambda_2 A}{B} \sigma_1 \sigma_2 \rho \right),
\label{Sig2}
\end{eqnarray}
with $A=E(X)$, $B=E(Y)$, $\sigma_1^2 = Var(X)$, $\sigma_2^2 = Var(Y)$ and $\rho = Cov(X,Y)/(\sigma_1 \sigma_2)$.
\end{itemize}

\begin{theorem} Based on what we defined above:
\begin{itemize}
\item[1.] If $A^2 \sigma_2^2 \geq B^2 \sigma_1^2$ the minimum variance $\Sigma_{min}= \Sigma_1(\lambda_1^{min})$,
with
\begin{eqnarray*}
\lambda_1^{min} = \frac{1}{2} + \frac{B \sigma_1 \rho}{2 A \sigma_2},\quad \Sigma_1 \hspace{1mm} given \hspace{1mm} in
\hspace{1mm} (\ref{Sig1}).
\end{eqnarray*}
\item[2.] If $A^2 \sigma_2^2 \leq B^2 \sigma_1^2$ the minimum variance $\Sigma_{min}= \Sigma_2(\lambda_2^{min})$,
with
\begin{eqnarray*}
\lambda_2^{min} = \frac{1}{2} + \frac{A \sigma_2 \rho}{2 B \sigma_1},\quad \Sigma_2 \hspace{1mm} given \hspace{1mm} in
\hspace{1mm} (\ref{Sig2}).
\end{eqnarray*}
\end{itemize}
\label{theo3}
\end{theorem}

\begin{proof}
We almost proved this theorem, indeed, one can easily verify that
$\lambda_1^{min}$ is the minimum of $\Sigma_1 (\lambda_1 )$ and $\lambda_2^{min}$ is the minimum
of $\Sigma_2 (\lambda_2 )$. To conclude we verify that $\Sigma_1 (\lambda ) \leq \Sigma_2 (\lambda ) $
if and only if $A^2 \sigma_2^2 \geq B^2 \sigma_1^2$.
\end{proof}

What is really appealing, in this theorem, is the fact that even if $\rho = 0$,
one should use $N = (1/2) N'$ or $N' = (1/2) N$ depending on whether $A^2 \sigma_2^2 \geq B^2 \sigma_1^2$ or not. Nevertheless,
in order to apply the results of either this theorem or Theorem \ref{theo4}, we should have a "sufficiently good" approximation
of $\sigma_1$, $\sigma_2$, $A$, $B$ and $\rho $. With our model $B = T_{s,t}[1 ](x)$ is explicitly known
and we can have $\sigma_2$ in the same fashion as $T_{s,t}[1 ](x)$. In section \ref{sec5}, procedure \textbf{P2} is implemented
by using the closed expression of $B$ and $\sigma_2$ and simulating $\sigma_1$, $A$, $\rho $ to get an approximation of
$\lambda_1^{min}$ or $\lambda_2^{min}$ that we use to re-simulate $Q$. In the case where $B$ and $\sigma_2$ are not known,
we can implement one of the two methods that are also efficient:
\begin{itemize}
  \item[\textbf{M1)}] Using all the simulated paths $N_{max}$, we approximate the values of
$\sigma_1$, $\sigma_2$, $A$, $B$ and $\rho $ then we compute $\lambda_1^{min}$ or
$\lambda_2^{min}$ that we use to re-simulate $Q$.
  \item[\textbf{M2)}] A fixed point alike method: Using all the simulated paths $N_{max}$, we approximate the values of
$\sigma_1$, $\sigma_2$, $A$, $B$ and $\rho $ then fix a threshold $\epsilon $ and test the condition
$A^2 \sigma_2^2 \geq B^2 \sigma_1^2$:
\begin{itemize}
\item[If] $A^2 \sigma_2^2 \geq B^2 \sigma_1^2$: Use the previous approximations except $A$ that will be
simulated using $\lambda_1 N_{max}$ paths, such that $\lambda_1^{min}$ is reached when
\begin{eqnarray*}
\left| \lambda_1 - \frac{1}{2} - \frac{B \sigma_1 \rho}{2 A \sigma_2} \right| < \epsilon.
\end{eqnarray*}
\item[If] $A^2 \sigma_2^2 \leq B^2 \sigma_1^2$: Use the previous approximations except $B$ that will be
simulated using $\lambda_2 N_{max}$ paths, such that $\lambda_2^{min}$ is reached when
\begin{eqnarray*}
\left| \lambda_2 - \frac{1}{2} - \frac{A \sigma_2 \rho}{2 B \sigma_1} \right| < \epsilon.
\end{eqnarray*}
\end{itemize}
\end{itemize}
\vspace{10mm}
In the following theorem, we will answer on whether we should implement the simulation procedure \textbf{P1} or \textbf{P2}.
\begin{theorem}
Based on what was defined above and on the values of $\lambda_1^{min}$ and $\lambda_2^{min}$ given in Theorem \ref{theo3}, if\\
1. $A^2 \sigma_2^2 \geq B^2 \sigma_1^2$ and $
1 \geq \rho > \frac{A \sigma_2}{B \sigma_1}\left( \frac{ \sqrt{13} -3}{2} \right) $
then $\left( B^2 \Sigma_1(\lambda_1^{min}) - \sigma_1^2 \right) < 0$.\\
2. $A^2 \sigma_2^2 \leq B^2 \sigma_1^2$ and $
1 \geq \rho > \frac{B \sigma_1}{A \sigma_2}\left( \sqrt{\frac{5}{4} + \frac{2 A^2 \sigma_2^2}{B^2 \sigma_1^2}} - \frac{3}{2} \right)$
then $\left( B^2 \Sigma_2(\lambda_2^{min}) - \sigma_1^2 \right) < 0$.
\label{theo4}
\end{theorem}
\begin{proof}
\begin{itemize}
\item[1.] If $A^2 \sigma_2^2 \geq B^2 \sigma_1^2$: $\Sigma_1(\lambda_1^{min}) = \Sigma_1(\rho)$
then we look for the condition on $\rho$ that allows that the trinomial
$\left( B^2 \Sigma_1(\rho) - \sigma_1^2 \right) $ is negative.
\item[2.] We go through the same argument as in 1.
\end{itemize}
\end{proof}
Theorem \ref{theo4} tells us that, even though we can compute explicitly the expression of $T_{s,t}[1 ](x)$, according to the correlation,
one can accelerate the convergence when using the quotient of two Monte Carlo estimators.

\section{Simulation and numerical results\label{sec5}}

In this section we test our simulations on a geometric average payoff that has the following payoff
\begin{eqnarray}
\Phi_{geo}^d (S_{T}) = \left( K - \prod^{d}_{i=1} (S_{iT} )^{1/d}
\right)_{+}.  \label{Putput}
\end{eqnarray}
In addition, we will test the American put on minimum and the American call
on maximum that have the following payoffs
\begin{eqnarray}
\Phi_{min} (S_{T}) = \left( K - \min (S^1_{T}, S^2_{T})
\right)_{+} ,\quad \Phi_{max} (S_{T}) = \left( \max (S^1_{T},S^2_{T}) -K
\right)_{+}. \label{Putcall}
\end{eqnarray}
The parameters of the simulations are the following: The strike $K=100$,
the maturity $T=1$, the risk neutral interest rate $r = \ln (1.1)$,
the time discretization is defined using the time steps that is
given as a parameter in each simulation, $S_{0}^{i}=100$ and
$\sigma_{ij}(t) = \sigma_{ij}(t) \delta (j-i)$ with $\sigma_{ii} = 0.2$.
The model considered is a multidimensional log-normal model
\begin{eqnarray*}
\frac{dS_{t}^i}{S_{t}^i} = r dt + (\sigma_{i} )' dW_t,\quad S_{i0} = y_{i} ,\quad i=1,..,d.
\end{eqnarray*}
All the prices and the standard deviations are computed using a sample of $16$ simulations. Besides,
the true values, to which we compare our simulation results, are set
using:
\begin{itemize}
\item the one-dimensional equivalence and a tree method \cite{bbsr}, available in Premia \cite{prem}, for options with $\Phi^n_{geo} (S_{T})$ as payoff,
\item the Premia implementation of a finite difference algorithm \cite{viln} in two dimensions for $\Phi_{min} (S_{T})$ and $\Phi_{max} (S_{T})$.
\end{itemize}

In Figure \ref{tableMCMLS}, we compare the \textbf{P2} ($N \neq N'$) version of MCM  with a standard
LS algorithm. The LS is implemented using linear regression for multidimensional contracts
and using up to three degree monomials for the one-dimensional contract. The reason behind
the choice of linear regression in the multidimensional case is the fact that the regression
phase of LS can really increase the execution time without a significant
amelioration of the prices tested.

In Figure \ref{tableMCMLS}, even if all the prices are sufficiently good,
we see that MCM provides better prices than those of LS. Also
when we increase the time steps, MCM is more stable than LS. However, for $n=10$
and time steps $> 10$, we remark that one should simulate $2^{14}$ trajectories to stablize MCM. This
fact is expected due to the important variance of the ten dimension contract and that one should simulate
more trajectories, on the one hand, to have an asymptotically good approximation of the relation between $N$ and $N'$
and, on the other, to have a sufficient number of trajectories for the approximation of the continuation.
The executions of MCM and LS with $2^{10}$ trajectories are carried out in less than one second.
Moreover, using $2^{14}$ trajectories the LS and MCM are executed within seconds ($<5s$).
As a conclusion from this figure, MCM provides better results than LS in approximately the same
execution time. When we increase the simulated trajectories to $2^{14}$, the MCM prices
are stabilized for high dimensions and are always better than LS prices.
\begin{figure}[H]
\begin{center}
\scalebox{0.7} {\hspace{-13mm}\includegraphics{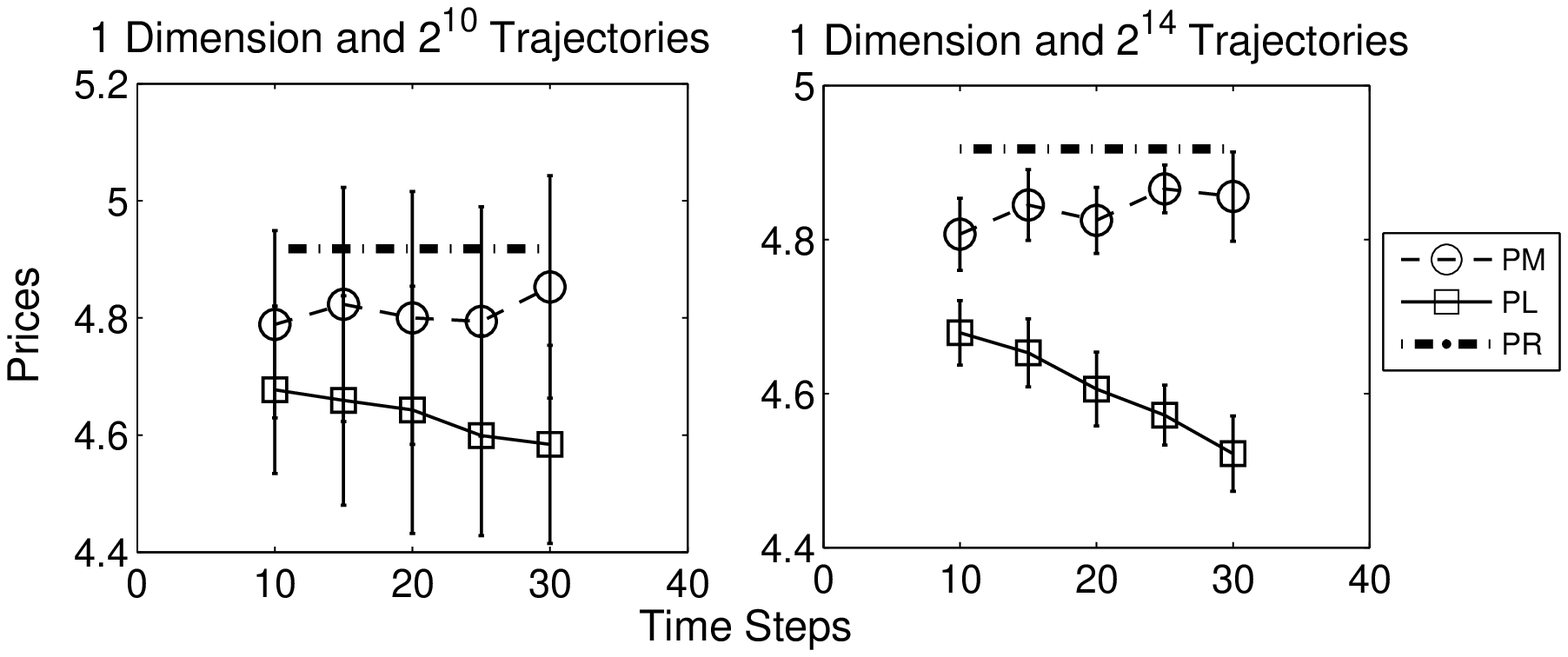}}
\vspace{-4mm}
\scalebox{0.7} {\hspace{-13mm}\includegraphics{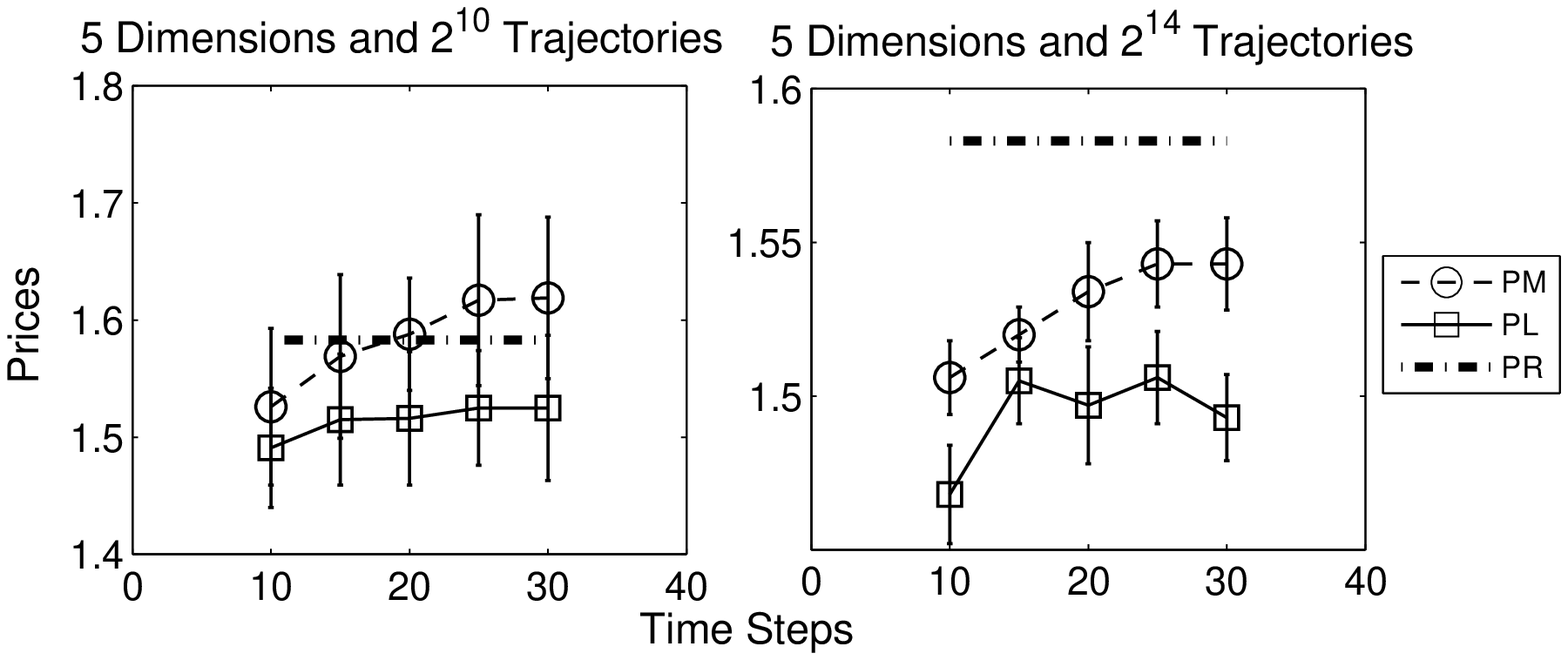}}
\vspace{-4mm}
\scalebox{0.7} {\hspace{-13mm}\includegraphics{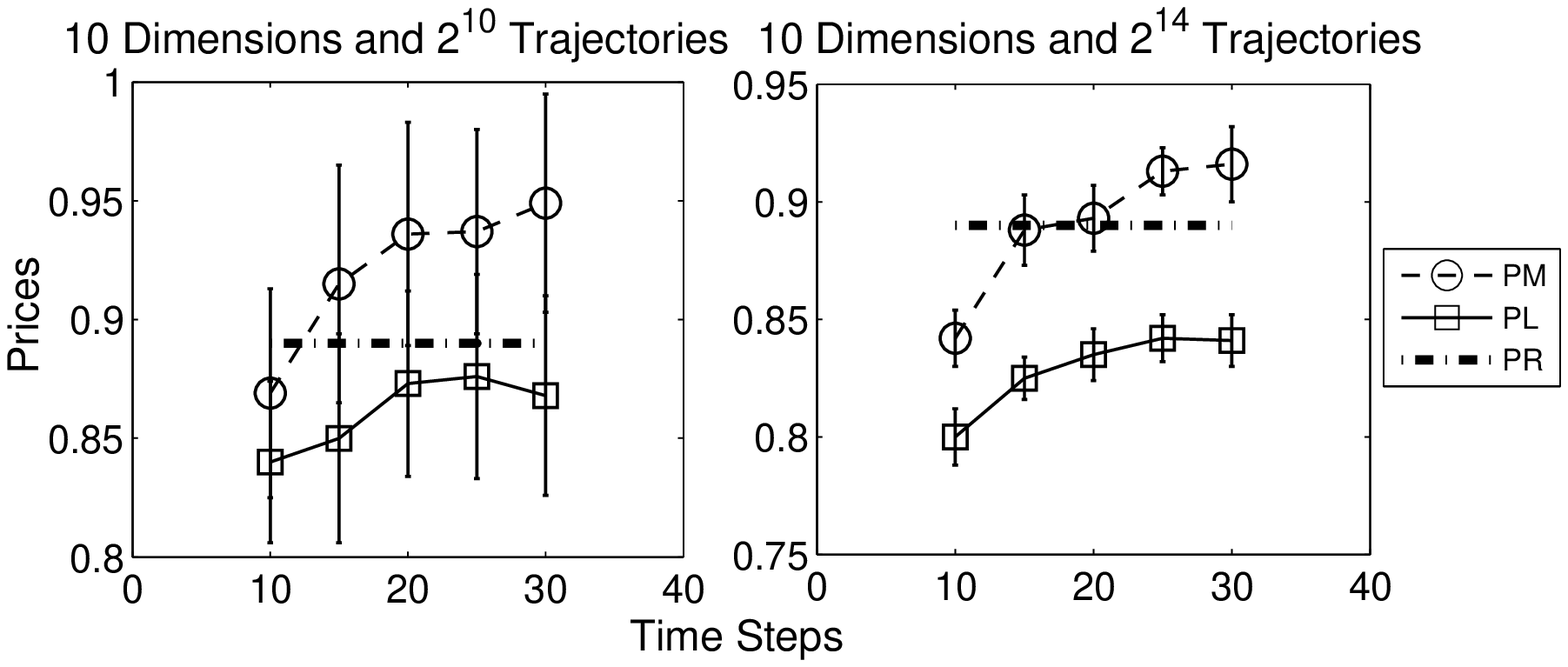}}
\caption{\label{tableMCMLS}{\it MCM Vs. LS for $\Phi_{geo}^d (S_{T}) $:
PR is the real price. PM and PL are the prices obtained
respectively by MCM and LS represented with their standard deviations.}}
\end{center}
\end{figure}
In Table \ref{tableP1P2}, we remain with the same payoff $\Phi_{geo}^d (S_T)$ but this
time we compare the different nonparametric methods of implementing MCM. In \textbf{P2($=$)}
and \textbf{P2(Opt)}, we use the same \textbf{P2} method but with $N = N'$ for the first
one and $N \neq N'$ for the second (The relation between $N$ and $N'$ is detailed in pages 16 and 17).
First, we remark that \textbf{P2($=$)} is not stable in the multidimensional
case and can give wrong results if the time steps $> 10$. However the \textbf{P2} method is stabilized
when we implement the version $N \neq N'$ of the advanced variance reduction
method detailed in section \ref{sec4}. Also when we use $2^{10}$ trajectories,
\textbf{P1} and \textbf{P2(Opt)} are almost similar. Nevertheless, with $2^{14}$ trajectories,
\textbf{P2(Opt)} outperforms \textbf{P1} which indicates that we fill the conditions of Theorem \ref{theo4}
and we have an asymptotically good approximation of the relation between $N$ and $N'$. As far as the execution
time is concerned, the time consumed by \textbf{P2(Opt)} is not much different from \textbf{P1} when we use $2^{10}$ trajectories. In addition, using $2^{14}$ trajectories, the computations of the relation between $N$ and $N'$ can be performed on the CPU when the rest of the simulation is done on the GPU. The latter fact allows a similar overall execution time for \textbf{P2(Opt)} and \textbf{P1} (within seconds).
\begin{table}[H]
\caption{\label{tableP1P2} \textbf{P1} Vs. \textbf{P2} for $\Phi_{geo}^d (S_{T}) $: The real
values are equal to $4.918$, $1.583$ and $0.890$ for dimensions
one, five and ten respectively}
\vspace*{1mm}
\hspace*{-3mm}
\begin{tabular}{|cccrrrrrr|}
\hline \multicolumn{1}{|c}{Simulated} & \multicolumn{1}{c}{Dim} & \multicolumn{1}{c}{Time}
&   \multicolumn{3}{c}{Price} & \multicolumn{3}{c|}{Std Deviation} \\ \multicolumn{1}{|c}{Paths} & \multicolumn{1}{c}{n} &
 \multicolumn{1}{c}{Steps} &  \multicolumn{1}{c}{\textbf{P1}} &  \multicolumn{1}{c}{\textbf{P2($=$)}}
& \multicolumn{1}{c}{\textbf{P2(Opt)}} & \multicolumn{1}{c}{\textbf{P1}} &  \multicolumn{1}{c}{\textbf{P2($=$)}}
& \multicolumn{1}{c|}{\textbf{P2(Opt)}} \\
\hline
$2^{10}$ & $1$ & $10$     &$4.750$ & $4.826$ & $4.789$    &$0.213$ & $0.167$ & $0.160$\\
$2^{10}$ & $1$ & $20$     &$4.729$ & $4.880$ & $4.800$    &$0.270$ & $0.226$ & $0.216$\\
$2^{10}$ & $1$ & $30$     &$4.679$ & $4.909$ & $4.853$    &$0.270$ & $0.179$ & $0.190$\\
\hline
$2^{10}$ & $5$ & $10$     &$1.548$ & $1.681$ & $1.526$    &$0.071$ & $0.073$ & $0.067$\\
$2^{10}$ & $5$ & $20$     &$1.632$ & $> 2.0$ & $1.588$    &$0.070$ & $     $ & $0.048$\\
$2^{10}$ & $5$ & $30$     &$1.650$ & $> 2.3$ & $1.619$    &$0.074$ & $     $ & $0.069$\\
\hline
$2^{10}$ & $10$ & $10$    &$0.900$ & $1.112$ & $0.869$    &$0.039$ & $0.045$ & $0.044$\\
$2^{10}$ & $10$ & $20$    &$0.921$ & $> 1.3$ & $0.936$    &$0.043$ & $     $ & $0.047$\\
$2^{10}$ & $10$ & $30$    &$0.908$ & $> 1.5$ & $0.949$    &$0.035$ & $     $ & $0.046$\\
\hline
\hline
$2^{14}$ & $1$ & $10$     &$4.738$ & $4.812$ & $4.807$    &$0.057$ & $0.046$ & $0.047$\\
$2^{14}$ & $1$ & $20$     &$4.675$ & $4.869$ & $4.825$    &$0.047$ & $0.044$ & $0.043$\\
$2^{14}$ & $1$ & $30$     &$4.638$ & $4.876$ & $4.856$    &$0.072$ & $0.059$ & $0.058$\\
\hline
$2^{14}$ & $5$ & $10$     &$1.487$ & $1.526$ & $1.506$    &$0.057$ & $0.012$ & $0.012$\\
$2^{14}$ & $5$ & $20$     &$1.504$ & $1.639$ & $1.534$    &$0.047$ & $0.021$ & $0.016$\\
$2^{14}$ & $5$ & $30$     &$1.508$ & $> 1.8$ & $1.543$    &$0.072$ & $     $ & $0.015$\\
\hline
$2^{14}$ & $10$ & $10$    &$0.845$ & $0.938$ & $0.842$    &$0.013$ & $0.015$ & $0.012$\\
$2^{14}$ & $10$ & $20$    &$0.901$ & $> 1.2$ & $0.893$    &$0.012$ & $     $ & $0.014$\\
$2^{14}$ & $10$ & $30$    &$0.923$ & $> 1.3$ & $0.916$    &$0.015$ & $     $ & $0.016$\\
\hline
\end{tabular}
\end{table}

Because of the bad results obtained previously with \textbf{P2($=$)}, we eliminate
this method and we only consider \textbf{P2(Opt)} and \textbf{P1}. In Table
\ref{tableP1P2MIN}, we analyze the American put on minimum and the American call
on maximum in two dimensions. As far as $\Phi_{min}$ is concerned,
\textbf{P2(Opt)} outperforms \textbf{P1} even when we use only $2^{10}$.Regarding $\Phi_{max}$, \textbf{P1} performs better than \textbf{P2(Opt)} for $2^{10}$ trajectories
which indicates that, because of the big variance produced by $\Phi_{max}$ relatively to
$\Phi_{min}$, the relation between $N$ and $N'$ is not well estimated.
Simulating $2^{14}$ trajectories, we obtain similar results for \textbf{P1} and \textbf{P2(Opt)}
for $\Phi_{max}$.

\begin{table}[t]
\caption{\label{tableP1P2MIN} \textbf{P1} Vs. \textbf{P2} for $\Phi_{min}$ and $\Phi_{max}$:
The real values are equal to $8.262$ and $21.15$ respectively}
\hspace{8mm}
\begin{tabular}{|cccrrrr|}
\hline \multicolumn{1}{|c}{Simulated} & \multicolumn{1}{c}{The } & \multicolumn{1}{c}{Time}
&   \multicolumn{2}{c}{Price} & \multicolumn{2}{c|}{Std Deviation} \\ \multicolumn{1}{|c}{Paths} & \multicolumn{1}{c}{ Payoff} &
 \multicolumn{1}{c}{Steps} &        \multicolumn{1}{c}{\textbf{P1}} & \multicolumn{1}{c}{\textbf{P2(Opt)}}
& \multicolumn{1}{c}{\textbf{P1}} & \multicolumn{1}{c|}{\textbf{P2(Opt)}} \\
\hline
$2^{10}$ & $\Phi_{min}$ & $10$     &$7.734$ & $7.986$     &$0.190$ & $0.248$ \\
$2^{10}$ & $\Phi_{min}$ & $20$     &$7.618$ & $7.895$     &$0.257$ & $0.270$ \\
$2^{10}$ & $\Phi_{min}$ & $30$     &$7.564$ & $7.920$     &$0.224$ & $0.263$ \\
\hline
$2^{10}$ & $\Phi_{max}$ & $10$     &$21.03$ & $20.33$     &$0.66$ & $0.86$ \\
$2^{10}$ & $\Phi_{max}$ & $20$     &$20.46$ & $19.38$     &$0.61$ & $0.73$ \\
$2^{10}$ & $\Phi_{max}$ & $30$     &$19.73$ & $18.13$     &$0.73$ & $0.93$ \\
\hline
\hline
$2^{14}$ & $\Phi_{min}$ & $10$     &$7.755$ & $8.088$     &$0.058$ & $0.067$ \\
$2^{14}$ & $\Phi_{min}$ & $20$     &$7.584$ & $8.098$     &$0.098$ & $0.052$ \\
$2^{14}$ & $\Phi_{min}$ & $30$     &$7.467$ & $8.087$     &$0.082$ & $0.043$ \\
\hline
$2^{14}$ & $\Phi_{max}$ & $10$     &$20.96$ & $20.91$     &$0.09$ & $0.24$ \\
$2^{14}$ & $\Phi_{max}$ & $20$     &$20.58$ & $20.56$     &$0.16$ & $0.16$ \\
$2^{14}$ & $\Phi_{max}$ & $30$     &$20.36$ & $20.05$     &$0.15$ & $0.22$ \\
\hline
\end{tabular}
\end{table}

\begin{figure}[t]
\epsfysize=6cm \epsfxsize=10cm
\centerline{\epsffile{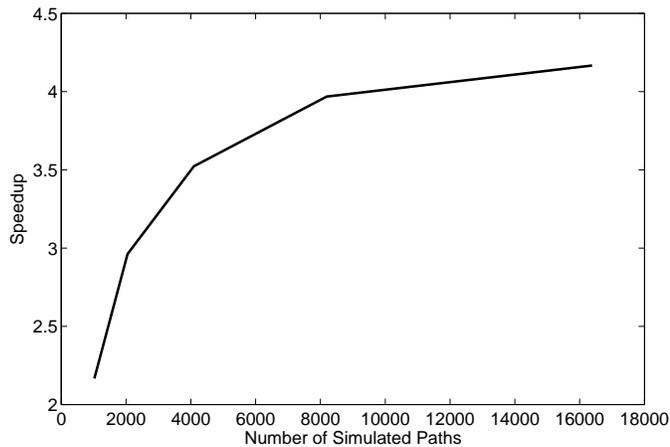}}
\caption{\label{figureCPU} The speedup of using all the CPU cores
according to the number of trajectories.}
\end{figure}

Let us now study the parallel adaptability of MCM for parallel architectures. In Figure
\ref{figureCPU}, we present the speedup of parallelizing\footnote{We use OpenMP directives.}
MCM on the four cores of the CPU instead of implementing it on only one core. We notice that
the speedup increases quickly according to the number of the simulated trajectories and it
reaches a saturation state for $> 9000$ trajectories. For a large dimensional problem, the maximum speedup obtained is
approximately equal to the number of logical cores on the CPU which indicates that MCM is very
appropriate for parallel architectures. We point out, however, that our parallelization
of MCM is done on the trajectories\footnote{which is the most natural procedure of parallelizing
Monte Carlo.}, so the speedup is invariable according to dimensions and time steps.

\begin{figure}[t]
\epsfysize=6cm \epsfxsize=10cm
\centerline{\epsffile{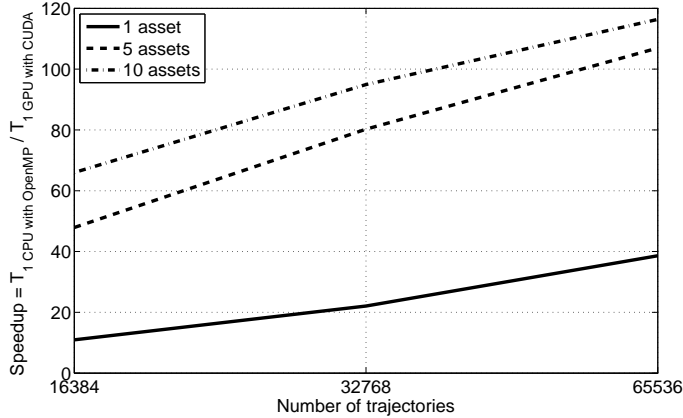}}
\caption{\label{figureGPU} The speedup of using the GPU instead of
the CPU cores according to the number of trajectories.}
\end{figure}

Regarding GPU implementation, we also use a path parallelization of simulations.
In Figure \ref{figureGPU}, we present the speedup of parallelizing\footnote{We use CUDA language.}
MCM on the GPU instead of implementing it on the four cores of the CPU. The speedup increases quickly
not only according to the number of simulated trajectories, but also according to the dimension
of the contract. The latter fact can be easily explained by the memory hierarchy of the GPU. The
speedups provided in Figure \ref{figureGPU} prove, once again, the high adaptability of MCM
on parallel architectures.

\section{Conclusion\label{sec6}}

In this article we provided, on the one hand, theoretical results that deal with
the continuation computations using the Malliavin calculus and how one can reduce the
Monte Carlo variance when simulating this continuation. On the other hand, we
presented numerical results related to the accuracy of the prices obtained and the
parallel adaptability of the MCM method on multi-core architectures.

As far as the theoretical results are concerned, based on the Malliavin calculus, we provided
a generalization of the value of the continuation for the multi-dimensional models with
deterministic and non a constant triangular matrix $\sigma (t)$. Moreover, we pointed out that
one can effectively reduce the variance by a simple conditioning method. Finally,
we presented a less intuitive but very effective variance reduction method based on
an appropriate choice of the number of trajectories used to approximate the quotient
of two expectations.

Regarding the numerical part, we proved that the one who looks for instantaneous simulations
can obtain better and sufficiently good prices with MCM than with LS using only $2^{10}$ trajectories.
Also, unlike LS, our nonparametric variance reduction implementation of MCM does
not require parametric regression. Thus we improve the results of the simulation
by only increasing the number of trajectories. Finally, increasing the number of trajectories
is time consuming but MCM can be effectively parallelized on multi-core CPUs and GPUs. Indeed,
the MCM simulation of $2^{14}$ trajectories using the GTX 480 GPU can be performed within seconds ($<5s$).

\section*{Appendix}

\preuve{of Lemma \ref{lem2}}{
The equality (\ref{SumSimp}) can be easily proved. Indeed, using the chain rule
\begin{eqnarray*}
D_{u}^{k} f(S_{t}) = \sum_{p=k}^{n} \sigma_{pk}(u) S_{t}^{p} \partial_{x_{p}}f(S_t)
\end{eqnarray*}
Besides, we assumed that $\rho (u) = \sigma^{-1}(u)$ which completes the proof.
}

\preuve{of Lemma \ref{lem1}}{
Using duality (\ref{dual1}) we have
\begin{eqnarray*}
\begin{array}{c}
E \left( h(S_{s}^k) F \sum_{i=k}^{n} \int_{I} \rho_{ik}(u) d W^{i}_u \right) =
E \left( \sum_{i=k}^{n} \int_{I} D_u^i \left[ h(S_{s}^k) F \right]  \rho_{ik}(u) du \right) \\ \\ =
E \left( h(S_{s}^k) \sum_{i=k}^{n} \int_{I} D_u^i F \rho_{ik}(u) du \right) +
E \left( F \sum_{i=k}^{n} \int_{I} h'(S_{s}^k) \sigma_{ki}(u) \rho_{ik}(u) S_{s}^k du \right)
\end{array}
\end{eqnarray*}
Moreover, the fact that $\sigma (u)$ and $\rho (u)$ are two triangular matrices
such that $\rho_{kk} (u) = 1/ \sigma_{kk} (u)$ simplifies the last term which can
be also rewritten using the Malliavin derivative
\begin{eqnarray*}
E \left( F \int_{I} h'(S_{s}^k) S_{s}^k du \right) =
E \left( F \int_{I} \frac{D_u^k h(S_{s}^k)}{\sigma_{kk}(u)} du \right)
\end{eqnarray*}
This provides the required result.
}
\medskip \noindent {\bf Acknowledgment:} We started this work in
the ANR GCPMF project, and it is supported now by CreditNext
project. The authors want to thank Professor Damien Lamberton for his
review of our work and Professor Vlad Bally for his valuable advice.

\bibliography{gpupricer-paper-biblio}

\begin{thebibliography}{10}

\bibitem{lok}
L.~A. Abbas-Turki, S.~Vialle, B.~Lapeyre, and P.~Mercier, ``High dimensional
  pricing of exotic european contracts on a {GPU} cluster, and comparison to a
  {CPU} cluster,'' {\em Parallel and Distributed Computing in Finance, in IEEE
  International Parallel \& Distributed Processing Symposium}, May 2009.

\bibitem{lokch}
L.~A. Abbas-Turki and B.~Lapeyre, ``American options pricing on multicore
  graphic cards,'' {\em IEEE The Second International Conference on Business
  Intelligence and Financial Engineering}, July 2009.

\bibitem{long}
F.~A. Longstaff and E.~S. Schwartz, ``Valuing {American} options by simulation:
  A simple least-squares approach,'' {\em The Review of Financial Studies},
  vol.~14, no.~1, pp.~113--147, 2001.

\bibitem{aa}
P.~Glasserman, {\em {Monte Carlo} Methods in Financial Engineering}.
\newblock Applications of Mathematics, Springer, 2003.

\bibitem{lamb}
E.~Cl\'ement, D.~Lamberton, and P.~Protter, ``An analysis of a least squares
  regression algorithm for {American} option pricing,'' {\em Finance and
  Stochastics}, vol.~17, pp.~448--471, 2002.

\bibitem{tsi}
J.~Tsitsiklis and B.~V. Roy, ``Regression methods for pricing complex
  {American-style} options,'' {\em IEEE Transactions on Neural Networks},
  vol.~12, no.~4, pp.~694--703, 2001.

\bibitem{zan}
V.~Bally, L.~Caramellino, and A.~Zanette, ``Pricing {American} options by
  {Monte Carlo} methods using a {Malliavin} calculus approach,'' {\em Monte
  Carlo Methods and Applications}, vol.~11, pp.~97--133, 2005.

\bibitem{ball}
V.~Bally and G.~Pag\`es, ``A quantization algorithm for solving
  multidimensional discrete-time optimal stopping problems,'' {\em Bernoulli},
  vol.~9, no.~6, pp.~1003--1049, 2003.

\bibitem{vlad}
V.~Bally, ``An elementary introduction to {Malliavin} calculus,'' {\em INRIA
  Rapport de Recherche}, vol.~4718, 2003.

\bibitem{bbsr}
M.~Broadie and J.~Detemple, ``American option valuation: new bounds,
  approximations, and a comparison of existing methods securities using
  simulation,'' {\em The Review of Financial Studies}, vol.~9, pp.~1221--1250,
  1996.

\bibitem{prem}
``http://www-roc.inria.fr/mathfi/{Premia}/,''

\bibitem{viln}
S.~Villeneuve and A.~Zanette, ``Parabolic {ADI} methods for pricing {American}
  options on two stocks,'' {\em Mathematics of Operations Research}, vol.~27,
  pp.~121--149, 2002.

\end{thebibliography}

\end{document}